\documentclass[12pt]{article}
\usepackage{amsfonts,amssymb,amsmath,amsthm,latexsym,a4wide,amsfonts,mathrsfs,bm,color,graphicx,amsthm,epsfig,multirow}
\usepackage{rotating}
\usepackage{lscape}
\usepackage{graphics}
\usepackage{mathrsfs,color}
\usepackage[round]{natbib}
\usepackage{booktabs}
\usepackage{setspace}
\usepackage{amsmath}
\usepackage{amsfonts}
\usepackage{amssymb}
\usepackage{graphicx,threeparttable}%
\usepackage{graphicx}
\usepackage{epstopdf}
\usepackage{epsfig}	
\usepackage{float}
\usepackage{amsfonts}
\usepackage{rotating}
\usepackage{fullpage}
\usepackage{setspace}
\usepackage{enumerate}
\usepackage{booktabs,threeparttable}
\usepackage{authblk}
\usepackage{subfigure}
\usepackage[dvipsnames]{xcolor}

\newtheorem{theorem}{Theorem}
\newtheorem{condition}{Condition}

\def\var{\hbox{$\mathrm{var}$}}
\newcommand{\e}{\mathbb{E}}
\newcommand{\tr}{\textrm{tr}}

\newcommand{\n}{\textrm{N}}

\newcommand{\hqed}{\hfill\qed}

\allowdisplaybreaks

\def\A{{\bf A}}

\def\B{{\bf B}}
\def\V{{\bf V}}

\def\N{\mbox{ $\mathcal{N}$}}

\def\bc{{\bf c}}
\def\w{{\bf w}}
\def\X{{\bf X}}
\def\x{{\bf x}}
\def\I{{\bf I}}

\def\Q{{\bf Q}}
\def\P{{\bf P}}
\def\E{{\mathbb{E}}}
\def\Ra{{\cal{D}}}

\def\Y{{\bf Y}}

\def\Z{{\bf Z}}
\def\M{{\bf M}}

\def\u{{\bf u}}
\def\e{{\bf e}}

\def\bmu{{\boldsymbol{\mu}}}

\def\blam{{\boldsymbol{\lambda}}}
\def\bthe{{\boldsymbol{\theta}}}

\def\bSig{{\boldsymbol{\Sigma}}}
\def\bLam{{\boldsymbol{\Lambda}}}
\def\beps{{\boldsymbol{\epsilon}}}
\def\bOme{{\boldsymbol{\Omega}}}
\def\bet{{\boldsymbol{\eta}}}
\def\brho{{\boldsymbol{\rho}}}
\def\biota{{\boldsymbol{\iota}}}
\def\E{\mathbb{E}}
\def\supw{\underset{\w\in \mathcal H}{\sup}}
\def\supwd{\underset{\w\in \mathcal H,d\in\Ra}{\sup}}
\def\infw{\underset{\w\in \mathcal H}{\inf}}

\def\calA{{\cal{A}}}

\def\calJ{{\cal{J}}}
\def\calH{{\cal{H}}}

\def\calS{{\cal{S}}}

\def\calT{{\cal{T}}}

\def\wh{\widehat}
\def\wt{\widetilde}

\def\n{\nonumber}
\def\be{\begin{eqnarray}}
\def\ee{\end{eqnarray}}

\def\sumt{\sum_{t\in\calT_0}}
\def\sumT{\sum_{t\in\calT_1}}
\def\sumi{\sum_{i=1}^J}
\def\sumj{\sum_{j=1}^J}

\pagebreak[4]

\begin{document}
	\begin{singlespace}
		\title{Asymptotic Properties of the Synthetic Control Method\thanks{We thank Stephane Bonhomme, Dick van Dijk, Yagan Hazard, and participants of the 16th International Symposium on Econometric Theory and Applications, the 2nd International Econometrics PhD conference, and the seminar at Erasmus University Rotterdam for valuable comments.}}
		
\author{Xiaomeng Zhang\footnote{Academy of Mathematics and Systems Science, Chinese Academy of Sciences},\quad Wendun Wang\footnote{Econometric Institute, Erasmus University Rotterdam; Tinbergen Institute}, \quad Xinyu Zhang\footnote{Academy of Mathematics and Systems Science, Chinese Academy of Sciences}}

		\date{\today}
		\maketitle
		
\begin{abstract}
\noindent 
This paper provides new insights into the asymptotic properties of the synthetic control method (SCM). We show that the synthetic control (SC) weight converges to a limiting weight that minimizes the mean squared prediction risk of the treatment-effect estimator when the number of pretreatment periods goes to infinity, and we also quantify the rate of convergence. Observing the link between the SCM and model averaging, we further establish the asymptotic optimality of the SC estimator under imperfect pretreatment fit, in the sense that it achieves the lowest possible squared prediction error among all possible treatment effect estimators that are based on an average of control units, such as matching, inverse probability weighting and difference-in-differences. The asymptotic optimality holds regardless of whether the number of control units is fixed or divergent. Thus, our results provide justifications for the SCM in a wide range of applications. The theoretical results are verified via simulations.
		\end{abstract}
	\end{singlespace}
	\thispagestyle{empty}
	
	\newpage
	\section{Introduction}
	The synthetic control method (SCM), proposed by
        \citet{abadie2003} and \citet{abadie2010}, has become one of
        the most popular approaches for policy evaluation. 
        The idea of the SCM is to
        construct a synthetic control (SC) unit intended to mimic the
        behavior of pretreatment outcomes of the treated unit to the greatest possible, such that one can use the posttreatment outcome of the synthetic control as the counterfactual of the treated outcome. Thus, in a seminal paper, \citet{abadie2010} regards a good pretreatment fit as a prerequisite for the standard SCM to work. 
        %
        Despite its wide range of applications and generalizations, the theoretical (asymptotic) properties of the SCM, especially under imperfect pretreatment fit, are relatively less studied; see \citet{abadie2021jel} for a review.
	
	Assuming that outcomes are generated from a linear factor model, \cite{ferman2021QE} investigates the asymptotic behavior of the SC weight when the pretreatment fit is not perfect. They demonstrate that the SC weight converges to a limit that generally does not recover the factor structure of the treated unit as the number of pretreatment periods increases, implying an asymptotic bias of the SC estimate of the treatment effect. 
	\cite{ferman2021} further examines the case of imperfect pretreatment fit when the number of control units goes to infinity and derives the conditions under which the factor structure of the treated unit can be asymptotically recovered by the SC unit, such that the SC estimator is asymptotically unbiased. A crucial condition here is that there exist weights that are diluted among an increasing number of control units and at the same time can produce an SC unit whose factor structure asymptotically reconstructs that of the treated unit. While the analysis is very informative, the existence of such weights is not automatically guaranteed. It is not yet clear in which cases these presumed diluting weights exist and which quantities affect the convergence of the SC weight. Moreover, it is also unclear how the SC estimator performs compared to other popularly used treatment-effect estimators when the pretreatment fit is imperfect.
	
	In this paper, we provide new insights into the asymptotic properties of the SCM. First, we show that the SC weight
    converges to a limiting weight that minimizes the mean squared prediction risk of the treatment-effect estimator, and we also
    quantify the rate of convergence. Under a widely studied linear factor model \citep[see, e.g.,][]{abadie2010,hsiao2019,botosaru2019}, we demonstrate that our limiting weight is compatible with the infeasible weight studied in \cite{ferman2021QE}, and thus it also
    balances the two parts of errors in fitting the factor
    structure and idiosyncratic shocks. This result confirms the
    finding in \citet{ferman2021QE} that the SC estimator
    is asymptotically biased under imperfect pretreatment fit,  and is also in line with the argument of \citet{bottmer2021} that the SC estimator is generally biased in a design-based framework. The property of being biased while asymptotically reaching the minimum mean squared prediction risk suggests that there is a bias-variance tradeoff in the SC estimator under imperfect pretreatment fit.  
    We also complement \citet{ferman2021QE} by quantifying the rate of
    convergence. We find that a better pre- and posttreatment fit
    both facilitate the convergence of the SC weight as expected. The role of the number of pretreatment periods  is mixed, depending on the goodness of fit before and after the treatment. We also find that a larger number of control units is associated with a slower convergence rate.
        
    Second, our derivation also offers a method to verify the existence of the diluting weight assumed in \cite{ferman2021}. We provide a sufficient condition under which our limiting weight is diluted and can asymptotically reconstruct the factor structure of the treated unit when the number of control units diverges. Intuitively, it requires that the information contained in the factor structure neither diminishes nor diverges. Nondiminishing guarantees that the factor structure of the treated unit can be reconstructed by the SC unit and nondivergence guarantees that weights can dilute among control units.
	
	Finally, motivated by the bias-variance tradeoff of the SC estimator, we explore its (expected) mean squared prediction error (MSPE), namely the (expected) mean squared loss of the SC estimator in the posttreatment periods.\footnote{We emphasize the prediction aspect of errors to highlight the out-of-sample feature of the treatment-effect estimate as it employs pretreatment information to extrapolate posttreatment counterfactual outcomes.} We show that under imperfect pretreatment fit\footnote{Here an imperfect pretreatment fit refers to the situation where there are no weights such that the outcome and observed covariates of the SC unit precisely equal those of the treated unit at each pretreatment period. Further discussion on the definition of imperfect fit is provided below Condition~\ref{as:xi1}.}, the SC estimator is asymptotically optimal in the sense that it achieves the lowest possible prediction risk (and loss) among all possible treatment effect estimators that are based on a (weighted) average of control units, when the number of pretreatment periods goes to infinity. Importantly, the asymptotic optimality holds regardless of
    whether the number of control units is fixed or divergent. In fact, under a fixed number of control units, the SCM makes the optimal tradeoff between bias and variance. With an increasing number of control units, the SC estimator is asymptotically unbiased, and thus the variance of the SC estimator converges to its lower bound. 
	Note that there are several treatment effect estimators that are based on a (weighted) average of the control units, such as the matching estimator, the inverse probability weighting (IPW), and difference-in-differences (DID) when an intercept is introduced \citep{doudchenko2016}, our result implies that these methods cannot outperform the SCM in terms of MSPE under imperfect pretreatment fit, at least asymptotically. The asymptotic optimality of SCM provides theoretical foundation for the numerical finding of \citet{bottmer2021} that SC estimators produce lower root mean squared errors than difference-in-means or DID in their simulation studies. 
	It is worth noting that while we base on linear factor models to demonstrate the main results, the asymptotic optimality continues to hold in a model-free setup, namely without assuming that outcomes are generated by a linear factor model (see Section~\ref{sec:general-model}). This generalization alleviates the concern that linear factor models are not sufficiently general to depict the data generating process (DGP) of potential outcomes, and provides guarantee for the SCM in a broad range of situations.

	Our analysis offers a justification for the SCM under imperfect pretreatment fit. \citet{abadie2010} requires perfect pretreatment fit to guarantee the unbiasedness of the SC estimator. \cite{botosaru2019} relaxes the perfect fit of covariates and shows that the asymptotic unbiasedness still holds as long as the pretreatment fit of outcomes is perfect. This relaxation comes at the cost of imposing stronger assumptions on the effects of covariates. \cite{ferman2021QE} shows that the SC estimator is
    biased under imperfect pretreatment fit, and \cite{ferman2021} further argues that this bias is diminishing as the number of
        control units increases. Our result of asymptotic optimality
        of the SCM does not rely on the divergence of the number of
        control units or the effect of covariates. It suggests that
        although the SC estimator may be biased under imperfect
        pretreatment fit, it is still the best choice among all other
        estimators of a similar construction. Considering the fact that an absolutely perfect fit is hardly achievable in real data, our result significantly widens the applicability of the SCM.

	Our analysis of asymptotic optimality is inspired by and contributes to the optimal model averaging literature
        \citep[see, e.g.,][among many others]{hansen2007, wan2010,
          liu2013, zhang2021}, which seeks to achieve the best
        prediction by optimally combining estimators obtained from
        candidate models (with different specifications). Observing
        the link between the SCM and optimal model averaging, we build
        upon the asymptotic optimality of averaging estimators to examine the properties of the SC estimator. Unlike model averaging that combines the estimates of candidate models, the SCM synthesizes  the realized (posttreatment) observations of the control units. Thus, our analysis does not constrain the behavior of candidate controls.
	An important contribution of this paper to the model averaging literature is that we provide the first attempt to prove the asymptotic optimality of the out-of-sample \emph{prediction} risk, which accounts for the randomness of both outcomes and weights. 
	In contrast, the majority of optimal averaging studies examine the risk while assuming that the weights are fixed.
	The only exception is \cite{zhang2021}, who analyze the risk accounting for the randomness of weights, but they only consider the in-sample risk. 
	
	In an independent and parallel study, \citet{chen2022} associates SCM with online learning and shows that SCM can perform almost as well as the best weighted average in a worse-case scenario. While their overall conclusion is somewhat related with our asymptotic optimality, 
	we differ from this study from several major perspectives. First, to link with online learning, \cite{chen2022} considers a thought experiment, where the outcomes are generated by an adversary and become sequentially available over time, such that the SC weights are calculated in a time-varying manner. While the adversarial framework is useful in several senses, it restricts the analysis to one-step-ahead prediction of the outcome, which is regarded as an undesirable departure from the standard SCM as the author points out in his conclusion. In contrast, we consider the standard setup of SCM with pretreatment outcomes fully accessible for prediction. This framework facilitates multiple-step-ahead prediction, and thus allows estimation of the treatment effect in multiple posttreatment periods simultaneously, a more common practice for SCM. Second, the targeting best weight in \citet{chen2022} is defined to minimize the mean squared loss over both pre- and posttreatment periods. Such an objective function also deviates from the goal of SCM (and other methods for treatment evaluation) that focuses on only the posttreatment prediction accuracy. With the pretreatment fit also accounted for in the evaluation, the overall regret guarantee in \citet{chen2022} does not necessarily imply the minimum loss in the posttreatment periods. In contrast, our asymptotic optimality explicitly concerns the mean squared loss in the posttreatment periods. Finally, \cite{chen2022} illustrates the advantage of SCM in terms of the regret bound,  while our asymptotic optimality concerns the ratio of (expected) MSPE over its infimum. Under certain regularity conditions, the asymptotic optimality can imply the bound convergence of the corresponding regret (see the end of Section~\ref{sec:opt-no-intercept} for details).  

	The remainder of this paper is organized as follows. Section
    \ref{sec:model framework} describes the model setup. Section
        \ref{sec:theoretical results} presents the main theoretical
        results, where we examine the convergence and establish the
        asymptotic optimality. The theory is verified via simulations
        in Section \ref{sec:simulation}, and Section
        \ref{sec:discussion} concludes the paper. The proofs of the main theorems are provided in the Appendix. The Online Appendix contains additional theoretical results.

	\section{Model framework} 
	\label{sec:model framework}
	
	Our framework considers the canonical SCM panel data, where there are $i\in\{0,1,\ldots,J\}$ units with $i=0$ being the only treated unit and the remaining $\{1,\ldots,J\}$ being the control units. The set of potential control units is also referred to as the ``donor pool''.
	Suppose that the treatment is assigned after time 0; then, we denote $\calT_0=\{-T_0+1,\ldots,0\}$ and $\calT_1=\{1,\ldots,T_1\}$ as $T_0$ pretreatment and $T_1$ posttreatment periods, respectively.
	Potential outcomes are denoted by $y_{i,t}^I$ when unit $i$ is
    treated at time $t$ and by $y_{i,t}^N$ when not treated.

	Following~\citet{ferman2021QE}, we assume that the potential outcomes are generated from a factor structure, i.e., 
	\be
	y_{i,t}^N&=&\bthe^{\top}_t\Z_i+\blam^{\top}_t\bmu_i+\delta_t+c_i+\epsilon_{i,t},\label{DGP-factor model}\\
	y_{i,t}^I&=&\alpha_{i,t}+y_{i,t}^N,
	\label{DGP-treatment}
	\ee
	where $\blam_t=\left(\lambda_{1,t},\ldots,\lambda_{F,t}\right)^\top$ is an $F\times1$ vector of unobserved common factors with unknown factor loadings $\bmu_i=\left(\mu_{i,1},\ldots,\mu_{i,F}\right)^\top$ $(F\times1)$, 	$\Z_i$ is an $r\times1$ vector of observed time-invariant covariates that are not affected by the treatment, $\bthe_t$ $(r\times1)$ represents the associated unknown parameters, $\delta_t$ can be interpreted as an unknown common factor with homogeneous loadings across units, $c_i$ captures the unobserved individual heterogeneity, and $\epsilon_{i,t}$ is the idiosyncratic shock. Here both $\delta_t$ and $c_i$ can be included in the linear factor structure $\blam^{\top}_t\bmu_i$; thus, our setup also encompasses  \citet{abadie2010}, \citet{botosaru2019}, \citet{Powell2022} and \citet{ferman2021}, among others. In this article, we treat $r$ and $F$ as fixed constants.
	We denote $\M=\left(\bmu_1,\ldots,\bmu_J\right)$, $\beps_t=\left(\epsilon_{0,t},\ldots,\epsilon_{J,t}\right)^\top$ and $\bc=\left(c_1,\ldots,c_J\right)^\top$. All matrices and vectors are marked in bold. The linear factor model seems a general and benchmark setup for most theoretical analysis of SCM and are particularly useful to illustrate some properties. Nonetheless, considering that the SCM may still be applicable even without specifying an outcome model, we shall relax this model assumption and analyze the properties of SCM in a model-free setup in Section~\ref{sec:general-model}.
	
	The interest is in estimating the treatment effect $\alpha_{i,t}$ of the single treated unit ($i=0$) for $t\in\mathcal{T}_1$. However, in practice, one cannot observed all potential outcomes but only the realized outcomes given as  
	\be
	y_{i,t}=d_{i,t}y_{i,t}^I+(1-d_{i,t})y_{i,t}^N, \quad i\in\{0,1,\ldots,J\}\ \textrm{and}\ t\in\calT_0\cup\calT_1,
	\label{observation}
	\ee 
	with $d_{i,t}$ being a dummy that equals 1  for the treated unit $i=0$ at $t\in\calT_1$, zero otherwise.
	The idea of the SCM is to construct an SC unit out of the donor pool, whose outcome serves as a proxy for the counterfactual outcome of the treated unit, i.e.,
	\be
	\wh\alpha_{0,t}(\wh\w)=y_{0,t}^I-\wh y_{0,t}^N(\wh\w)=y_{0,t}-\sumj \wh w_j y_{j,t}, \quad t\in\calT_1,
	\n
	\ee
	where $\wh\w$ is the SC weight determined by some predictors. We focus on the popular specification where the predictors include all pretreatment outcomes \citep[see, e.g.,][]{doudchenko2016,ferman2021QE,ferman2021}, and we also allow for observed covariates.
	Let $\X_i=\left(y^N_{i,-T_0+1},\ldots,y^N_{i,0},\Z_i^\top\right)^\top=\left(y_{i,-T_0+1},\ldots,y_{i,0},\Z_i^\top\right)^\top$ be a $(T_0+r)\times1$ vector of the pretreatment characteristics of unit $i$ for all $i$, and define $\boldsymbol{\mathcal{X}}_c=\left(\X_1,\ldots,\X_J\right)$ with subscript $c$ representing control units. The original SCM restricts the weights to be in the set $\calH_{\text{orig}}=\left\{\w=(w_1,\ldots,w_J)^\top\in[0,1]^J\mid\sumj w_j=1\right\}$, such that control units are combined convexly. To allow for possible extrapolation, we relax the positivity constraint and consider $\calH=\left\{\w=(w_1,\ldots,w_J)^\top\in[-C_\text{L},C_\text{U}]^J\mid\sumj w_j=1\right\}$, where $C_\text{L}$ and $C_\text{U}$ are two nonnegative constants. Clearly, $\calH\supseteq\calH_{\text{orig}}$, and thus our results also hold for the original SC weight. Note that since $\calH$ is bounded, the range of extrapolation it permits is bounded, and the sum-up-to-unity restriction still plays a crucial role.
	
	For any given positive definite matrix $\V$, the SC weight can be obtained by solving the following optimization:
	$$
	\wh\w(\V)=\underset{\w\in\mathcal{H}_{\text{orig}}}{\arg\min}\sqrt{\left(\X_0-\boldsymbol{\mathcal{X}}_c\w\right)^\top\V\left(\X_0-\boldsymbol{\mathcal{X}}_c\w\right)}.
	$$
	For the sake of simplicity, we follow \cite{ferman2021QE} and \cite{ferman2021} to set $\V=\I_{T_0+r}$, where $\I_{T_0+r}$ denotes the $(T_0+r)\times(T_0+r)$ identity matrix.
	Then, the SC weight can be obtained by
	\be
	\wh\w&=&\underset{\w\in\mathcal{H}}{\arg\min} L_{T_0}(\w)\n\\
	&=&\underset{\w\in\mathcal{H}}{\arg\min}\frac{1}{T_0}\left\|\X_0-\boldsymbol{\mathcal{X}}_c\w\right\|^2\n\\
	&=&\underset{\w\in\mathcal{H}}{\arg\min} \frac{1}{T_0}
	\left\{\sum_{t\in\calT_0}\left(y_{0,t}-\sumj w_j y_{j,t}\right)^2+\left\|\Z_0-\sumj w_j\Z_j\right\|^2\right\},
	\label{SC weight}
	\ee 
	where $\left\|\x\right\|=\sqrt{\x^\top\x}$ for any vector $\x$. 
	We first analyze the properties of the SC weight defined as above. In Section~\ref{sec:optimality with intercept}, we consider adding an intercept to~\eqref{SC weight}, such that the SCM can also be compared with the difference-in-differences (DID) estimator.

	\section{Theoretical results}\label{sec:theoretical results}
	
	This section presents the main theoretical results. Considering the fact that MSPE is widely used to evaluate the accuracy of treatment effect estimates, we focus on the (expectation of) MSPE and study how the SC estimator and weights are related to their optimal counterparts in terms of minimizing the MSPE.

	We first show that the SC weights converge to the infeasible optimal weights that minimize the risk of MSPEs. Our convergence
        results 
        suggest that the SCM generally leads to a biased treatment effect estimator but potentially with a reduced variance. We also complement the existing asymptotic analysis of the SC weight by quantifying the rate of convergence. Furthermore, we establish the asymptotic optimality of the SC estimator in the sense that it achieves the lowest possible loss among all possible averaging estimators of the control units.  
	Unless otherwise stated, all limiting properties hold when the number of pretreatment periods goes to infinity, i.e., $T_0\to\infty$.

	\subsection{Convergence of the SC weight}\label{sec:convergence}
	
	To show the convergence, we assume that the following regularity conditions hold. 
	\begin{condition}
		\label{as:data stochastic}\  
		\begin{enumerate}[(i)]
			\item \label{as:data fixed}
			We treat $\left\{ \bmu_i, c_i, \Z_i\mid i\in\{0,1,\ldots,J\} \right\}$, $\left\{ \blam_t, \delta_t\mid t\in\calT_0\cup\calT_1 \right\}$ and $\left\{\alpha_{0,t}\mid t\in\calT_1\right\}$ as fixed and $\left\{ \epsilon_{i,t}\mid i\in\{0,1,\ldots,J\}, t\in\calT_0\cup\calT_1 \right\}$ as stochastic.
			\item \label{as:expection of error term}
			$\E\epsilon_{i,t}=0$ for $i\in\{0,1,\ldots,J\}$ and $t\in\calT_0\cup\calT_1$.
		\end{enumerate}
	\end{condition}
	Condition~\ref{as:data stochastic} concerns the sampling of data. Condition~\ref{as:data stochastic}~\eqref{as:data fixed} is to simplify the proof, and a similar assumption is also used in \citeauthor{ferman2021QE} (2021, Assumption~2.1) and \citeauthor{ferman2021} (2021, Assumption~2). 
	Condition~\ref{as:data stochastic}~\eqref{as:expection of error term} requires the idiosyncratic shocks to have a zero-mean, which appears as a rather standard assumption in the SCM literature \citep[see, e.g.,][among others]{xu2017, botosaru2019, ferman2021, eli2021}. It can be interpreted as ``selection on unobservables'' because it allows for any form of dependence between the treatment assignment and the factor structure as long as the assignment is uncorrelated with idiosyncratic shocks \citep{ferman2021QE}.

	\begin{condition}
		\label{as:factor bound}\ 
		\begin{enumerate}[(i)]
			\item\label{lambda}	${T_1}^{-1}\sumT\left({T_0}^{-1}\sum_{k\in\calT_0}\blam_k^\top\blam_k-\blam_t^\top\blam_t\right)=O\left(T_0^{-1/2}\right).$
			\item \label{delta}  ${T_1}^{-1}\sumT\left({T_0}^{-1}\sum_{k\in\calT_0}\delta_k^2-\delta_t^2\right)=O\left(T_0^{-1/2}\right).$
			
		\end{enumerate}
	\end{condition}
	Condition~\ref{as:factor bound} requires that the variation of
        the common factors does not change substantially after
        treatment. This condition is in line with the rationality of the SCM since it implies that the main difference between the pre- and posttreatment outcomes is exclusively due to the treatment effect. Only in this case can the pretreatment data be used to construct the SC unit and estimate the treatment effect for the posttreatment periods.  
	If we treat $\left\{\blam_t\right\}$ and $\left\{\delta_t\right\}$ as stochastic and assume $T_1$ to be divergent at rate $O(T_0)$, then using the central limit theorem for dependent observations \citep[see, e.g., Theorems~5.16 and 5.20 in][]{white1984}, we can obtain that
	\begin{eqnarray}
	&&\frac{1}{\sqrt{T_0}}\sumt\blam_t^{\top}\blam_t={T_0}^{-1/2}\sumt\E\blam_t^{\top}\blam_t+O_p(1),  \ \frac{1}{\sqrt{T_0}}\sumT\blam_t^{\top}\blam_t={T_1}^{-1/2}\sumT\E\blam_t^{\top}\blam_t+O_p(1),\nonumber\\
	&&\frac{1}{\sqrt{T_0}}\sumt\delta_t^2={T_0}^{-1/2}\sumt\E\delta_t^2+O_p(1),\ \quad\qquad
	\frac{1}{\sqrt{T_0}}\sumT\delta_t^2={T_1}^{-1/2}\sumT\E\delta_t^2+O_p(1)\nonumber.
	\end{eqnarray} 
	In this case, Condition~\ref{as:factor bound} is satisfied in probability under certain stability assumptions, such as Assumptions~4--5 in \cite{ferman2021QE}.
	This condition also allows for ``divergent'' factors, i.e., $T_0^{-1}\sumt\blam_t\to\infty$. To see this more explicitly, consider a simple example where $T_1=T_0$, $F=1$ and $\lambda_{1,t}=|t-t_0|^{1/2}$ for $t\in\calT_0\cup\calT_1$, where $t_0$ denotes the time of treatment and $t_0=0$ with our notation.
	In this example, $T_0^{-1}\sumt\blam_t=T_0^{-1}\sum_{t=1}^{T_0}t^{1/2}$ goes to infinity, 
	and thus $\blam_{t}$ is a divergent factor.
	The condition is satisfied because ${T_1}^{-1}\sumT\left({T_0}^{-1}\sum_{k\in\calT_0}\blam_k^{\top}\blam_k-\blam_t^{\top}\blam_t\right)=0$. 
	Nevertheless, not all kinds of divergent factors satisfy this condition, e.g., $\lambda_{1,t}=t$. 

	\begin{condition}
		There exists a constant $C_0$ such that $\mu_{il}<C_0$ and $c_i<C_0$ for $i\in\{0,1,\ldots,J\}$ and $l\in\{1,\ldots,F\}$.
		\label{as:factor loading}
	\end{condition}
	Condition~\ref{as:factor loading} requires the uniform boundedness of the factor loadings and time-invariant fixed effects. 
	This condition easily holds under the often assumed conditions for factor identifiability, e.g., $\M\M^\top$ being a diagonal matrix \citep{xu2017} or $T_0^{-1}\M\M^\top\rightarrow\I_{F}$ \citep{stock2002}, where we recall that $\M=\left(\bmu_1,\ldots,\bmu_J\right)$. The same condition is also used in \citeauthor{ferman2021} (2021, Assumption~3.2(b)).  
	The uniform boundedness of $c_i$ is also a mild condition once we note that $c_i$ can be regarded as the loading of the factor $\lambda_t=1$.

	\begin{condition}
		\label{as:covariates}\ 
		\begin{enumerate}[(i)]
			\item \label{theta}
			${T_1}^{-1}\sumT\left({T_0}^{-1}\sum_{k\in\calT_0}\bthe_k^\top\bthe_k-\bthe_t^\top\bthe_t\right)=O(T_0^{-1/2}).$
			\item \label{Z}
			There exists a constant $C_z$ such that $\left\|\Z_i\right\|<C_z$ for $i\in\{0,1,\ldots,J\}$.
		\end{enumerate}
	\end{condition}
	Condition~\ref{as:covariates} concerns the variability of the observed covariates and their corresponding coefficients, and Conditions~\ref{as:covariates}~\eqref{theta} and~\eqref{Z} can be viewed as covariate versions of Conditions~\ref{as:factor bound} and~\ref{as:factor loading}, respectively. Specifically,  Condition~\ref{as:covariates}~\eqref{theta} requires that the change in the coefficients before and after the treatment is limited, and it obviously holds for any time-invariant coefficients. Condition~\ref{as:covariates}~\eqref{Z} requires uniformly bounded covariates across units. Both parts of Condition~\ref{as:covariates} can be justified in the same way as  
	Conditions~\ref{as:factor bound} and~\ref{as:factor loading}. Note that \cite{ferman2021} assumes that the observed covariates and their coefficients satisfy the same set of conditions for factors and loadings when he proves the results with covariates (see Section~A.2.5 of Supplementary Material in \cite{ferman2021}), which resembles our strategy here.

	We also need some restrictions on the relation between the idiosyncratic shock of the treated and control units. Let $e_{t,\epsilon}^{(i)}=\epsilon_{0,t}-\epsilon_{i,t}$ for $i\in\{1,\ldots,J\}$ and $t\in\calT_0\cup\calT_1$, 
	$R_{T_0}(\w)=\E L_{T_0}(\w)$ and $\xi_{T_0}=\inf_{\w \in \calH}R_{T_0}(\w)$. Intuitively, we can interpret $\xi_{T_0}$ as a measure of pretreatment fit.
	\begin{condition}
		\label{as:post error term bias} 
		$\supw\left|{T_0}^{-1}\sumt\E\left(\sumj w_j  e_{t,\epsilon}^{(j)}\right)^2-{T_1}^{-1}\sumT\E\left(\sumj w_j  e_{t,\epsilon}^{(j)}\right)^2\right|=o\left( \xi_{T_0}\right)$.
	\end{condition}
	This condition means that, with respect to the goodness of pretreatment fit, the difference in idiosyncratic shocks between the treated and any weighted average of control units does not change substantially after treatment. Generally, it is more likely to hold when the idiosyncratic shocks are more stable over time or more alike across units. For example, in the setting of \citet{ferman2021} where $\{\epsilon_{0,t},\epsilon_{1t},\ldots,\epsilon_{j,t}\}$ are of zero-mean and independent across units, we have 
	\be
	\E\left(\sumj w_j  e_{t,\epsilon}^{(j)}\right)^2
	=\E\left(\epsilon_{0,t}^2\right)+\sumj w_j^2\E\left(\epsilon_{j,t}^2\right),\quad t\in\calT_0\cup\calT_1.\n
	\ee
	If we further assume that the sequence $\{\epsilon_{i,t}\}_{t\in\calT_0\cup\calT_1}$ is stationary, then Condition~\ref{as:post error term bias} holds, because
	\be
	&&\supw\left| 
	\frac{1}{T_0}\sumt\E\left(\sumj w_j  e_{t,\epsilon}^{(j)}\right)^2-\frac{1}{T_1}\sumT\E\left(\sumj w_j  e_{t,\epsilon}^{(j)}\right)^2\right|
	\n\\
	&=&\supw\left| \frac{1}{T_0}\sumt\E\left(\epsilon_{0,t}^2\right)+\frac{1}{T_0}\sumt\sumj w_j^2\E\left(\epsilon_{j,t}^2\right)\right.
	\n\\
	&&\left.
	-\frac{1}{T_1}\sumT\E\left(\epsilon_{0,t}^2\right)-\frac{1}{T_1}\sumT\sumj w_j^2\E\left(\epsilon_{j,t}^2\right) \right|
	\n\\
	&=&0.
	\n
	\ee
	
	To state the next condition, define  $\Y_i=\left(y_{i,-T_0+1},\ldots,y_{i,0}\right)^\top$ for $i\in\{0,1,\ldots,J\}$,
	${\boldsymbol{\mathcal{Y}}_c}=\left(\Y_1,\ldots,\Y_J\right)$ and $\bSig=T_0^{-1}\E\left({\boldsymbol{\mathcal{Y}}_c}^\top{\boldsymbol{\mathcal{Y}}_c}\right)$, where the subscript ``$c$'' indicates control units. We use $\lambda_{min}(\cdot)$ and $\lambda_{max}(\cdot)$ to represent the minimum and maximum eigenvalue of a matrix.
	
	\begin{condition}
		\label{as:Y_c}
		There exist constants $\kappa_1$ and $\kappa_2$ such that
		$0<\kappa_1\leq\lambda_{min}\left(\bSig\right)\leq\lambda_{max}\left(\bSig\right)\leq\kappa_2$.
	\end{condition}
	This condition bounds the variability of the pretreatment
        outcomes of control units from both below and above. It can
        also be satisfied under many standard setups of the SCM. For example, consider pure factor models as the DGP, as in \cite{hsiao2019} and \cite{eli2021}, i.e., $y_{i,t}^N=\blam^{\top}_t\bmu_i+\epsilon_{i,t}$, and it can be written in a matrix form as ${\boldsymbol{\mathcal{Y}}_c}=\bLam^{-\top}\M+\beps^{-}$, where 
	$\bLam^{-}=\left(\blam_{-T_0+1},\ldots,\blam_{0}\right), \M=\left(\bmu_1,\ldots,\bmu_J\right)$
	and  $\beps^-=\left(\beps_{-T_0+1},\ldots,\beps_{0}\right)$ being a $T_0\times J$ matrix of idiosyncratic shocks with $\beps_t=\left(\epsilon_{0,t},\ldots,\epsilon_{J,t}\right)^\top$. In this case, we have $\bSig=T_0^{-1}\M^\top\bLam^{-}\bLam^{{-}\top}\M+T_0^{-1}\E\beps^{{-}\top}\beps^{-}$.
	Note
        that  $$\lambda_{min}\left(\A\right)+\lambda_{min}\left(\B\right)\leq\lambda_{min}\left(\A+\B\right)\leq\lambda_{max}\left(\A+\B\right)\leq\lambda_{max}\left(\A\right)+\lambda_{max}\left(\B\right),$$
        for any Hermite matrices $\A$ and $\B$ of the same
        order. Thus, we only need to study
        $\lambda_{min}\left(T_0^{-1}\M^\top\bLam^-\bLam^{-\top}\M\right)$
        and
        $\lambda_{min}\left(T_0^{-1}\E\beps^{-\top}\beps^{-}\right)$
        for the lower bound of $\lambda_{min}\left(\bSig\right)$ and study $\lambda_{max}\left(T_0^{-1}\M^\top\bLam^{-}\bLam^{{-}\top}\M\right)$ and $\lambda_{max}\left(T_0^{-1}\E\beps^{-\top}\beps^{-}\right)$ for the upper bound of $\lambda_{max}\left(\bSig\right)$.
	Under Condition~\ref{as:data stochastic}~\eqref{as:expection
          of error term},
        $\lambda_{min}\left(T_0^{-1}\E\beps^{-\top}\beps^{-}\right)$
        and
        $\lambda_{max}\left(T_0^{-1}\E\beps^{-\top}\beps^{-}\right)$
        are both finite as long as the idiosyncratic shock has a
        finite variance. For the other two terms involving factors and
        loadings, it is common to assume that
        $T_0^{-1}\bLam^{-}\bLam^{-\top}=\I_{F}$ when the factor is not
        divergent \citep[see, e.g.,][]{bai2009,xu2017,eli2021}. Thus, $\lambda_{min}\left( T_0^{-1}\M^\top\bLam^{-}\bLam^{-\top}\M \right)=\lambda_{min}\left( \M^\top\M \right)$ and $\lambda_{max}\left( T_0^{-1}\M^\top\bLam^{-}\bLam^{-\top}\M \right)=\lambda_{max}\left( \M^\top\M \right)$. Therefore, Condition~\ref{as:Y_c} can be simplified to stating that there exist constants $\wt\kappa_1$ and $\wt\kappa_2$ such that 
	\be
	0<\wt\kappa_1\leq\lambda_{min}\left(\M^\top\M\right)\leq\lambda_{max}\left(\M^\top\M\right)<\wt\kappa_2,
	\n
	\ee
	which is often used in factor models and resembles the rank condition in standard regressions \citep{bai2009,xu2017}. As for the diverging factors, Condition~\ref{as:Y_c} can be satisfied with a more restrictive assumption on the loadings than $T_0^{-1}\bLam^{-}\bLam^{-\top}=\I_{F}$.
	
	To evaluate the performance of the SC treatment effect estimator, we consider the MSPE for some weight $\w\in\calH$, defined as 
	\be
	L_{T_1}(\w)&=&\frac{1}{T_1}\sumT\left\{\alpha_{0,t}-\wh\alpha_{0,t}(\w)\right\}^2\n\\
	&=&\frac{1}{T_1}\sumT\left\{y_{0,t}^N-\wh y_{0,t}^N(\w)\right\}^2\\
	&=&\frac{1}{T_1}\sumT\left(y_{0,t}^N-\sumj  w_j y_{j,t}\right)^2,
	\n
	\label{LT1}
	\ee
	and its risk is $R_{T_1}(\w)=\E L_{T_1}(\w)$. The optimal
        weight vector for a given $T_1$ is defined as the minimizer of the risk, i.e.,
	\be
	\w^{\text{opt}}_{T_1}=\underset{\w\in\calH}{\arg\min} R_{T_1}(\w).
	\label{w0}
	\ee 
	
	\begin{theorem} \label{th:con}	
		Given any $T_1$, if $\w^{\text{opt}}_{T_1}$ is an interior point of $\calH$ and Conditions~\ref{as:data stochastic}--\ref{as:Y_c} hold, {then} 
		\begin{equation}
		\left\|\wh\w-\w^{\text{opt}}_{T_1}\right\|=O_p\left( T_0^{\nu}\xi_{T_0}^{1/2}+T_0^{\nu}\xi_{T_1}^{1/2}+T_0^{-1/4+\nu}J\right),
		\n
		\end{equation}
		where $\nu>0$ is a sufficiently small constant, $\xi_{T_0}=\inf_{\w \in \calH}R_{T_0}(\w)$, and $\xi_{T_1}=\inf_{\w \in \calH}R_{T_1}(\w)$.
	\end{theorem}
	Theorem~\ref{th:con} shows that as $T_0$ and $T_1\to\infty$, the SC weight $\wh\w$ converges to the optimal weight vector sequence $\w^{\text{opt}}_{T_1}$ at a rate that depends on $\xi_{T_0}$, $\xi_{T_1}$ and $J$.\footnote{The posttreatment periods $T_1$ plays a role in convergence via $\xi_{T_1}$.} We discuss the roles of $\xi_{T_0}$, $\xi_{T_1}$, $T_0$ and $J$ in turn. First, a faster rate of $\xi_{T_0}$ and $\xi_{T_1}$ going to zero implies quicker convergence of $\wh\w$. Recall that $\xi_{T_0}$ is a measure of pretreatment fit, and thus the theorem clearly links a good pretreatment fit with accurate weight estimation. 
	Second, the number of pretreatment periods $T_0$ plays a mixed
        role in the convergence rate, depending on the goodness of
        pre- and posttreatment fit. If the fit is good in both
        regimes, then the dominant term is $T_0^{-1/4+\nu}J$, and an
        increase in $T_0$ improves the accuracy of estimated
        weights. However, when the fit is poor, the net effect of
        increasing $T_0$ may be detrimental to weight
        convergence. This result is not surprising because a poor fit
        implies that the predictors are not informative for
        constructing the SC unit. Thus, increasing the sample size of uninformative or even misleading predictors drives the SC weight further from optimality.
	Finally, a larger $J$ is associated with a slower convergence rate. 
	Note that this result does not conflict with the asymptotic
        unbiasedness of the SC estimator when $J\to\infty$ shown by
        \citet{ferman2021}, because the SC weight may still converge to
        the unbiased limiting weight but just at a lower rate (see
        below for further detail on the relation with \citet{ferman2021}). Considering that $J$ corresponds to the number of weight parameters to be estimated, less accurate estimates are expected when the dimension of parameters increases. The negative role of $J$ in the rate of convergence is consistent with the conjecture of \citet{ferman2021} that stronger assumptions on the moments of $\epsilon_{i,t}$ are required when $J$ diverges at a faster rate.

	We discuss how Theorem~\ref{th:con} is related to the asymptotic result of \cite{ferman2021QE}.
	First, we show that the limit of the SC weight defined in~\eqref{w0} is compatible with the infeasible weight studied in \cite{ferman2021QE}, which we denoted as $\w^*_{\text{FP}}$ and satisfies $(\bmu_0^\top,c_0)^\top\neq(\M^\top,\bc)^\top\w^*_{\text{FP}}$. 
	To remain close to the benchmark setup of \citet{ferman2021QE}, we consider a simplified version of our DGP without observed covariates and assume that $T_0^{-1}\sumt\beps_t\beps_t^\top{\rightarrow}_p \sigma^2_{\beps}\I_{J+1}$ and $T_0^{-1}\sumt\blam_t\blam_t^\top{\rightarrow}_p \bOme_0$, where $\bOme_0$ is a positive semidefinite matrix.
	\citet{ferman2021QE} shows that the original SC weight ($\w\in\mathcal{H}_{\text{orig}}$) converges in probability to   $\w^*_{\text{FP}}$ that minimizes the following quantity
	$$Q_0(\w)=\left\{\left(c_0-\bc^\top\w\right)^2+\left(\bmu_0-\M\w\right)^\top\bOme_0\left(\bmu_0-\M\w\right)\right\}+\sigma^2_{\epsilon}\left(1+\w^\top\w\right).
	$$
	Recall that the limiting weight we study minimizes $R_{T_1}(\w)$. Thus, we examine how $R_{T_1}(\w)$ is related to $Q_0(\w)$. Denote $\beps_{c,t}=(\epsilon_{1,t},\ldots,\epsilon_{J,t})^\top$ as the error vector of control units. Then, $R_{T_1}(\w)$ can be decomposed as
	\be
	R_{T_1}(\w)&=&\left\{\left(c_0-\bc^\top\w\right)^2+\frac{1}{T_1}\sumT\left(\bmu_0-\M\w\right)^\top\blam_t\blam_t^\top\left(\bmu_0-\M\w\right)\right\}
	\n\\
	&&+\frac{1}{T_1}\sumT\E\left(\epsilon_{0,t}-\beps_{c,t}^\top\w\right)^2.
	\label{R1 decomposition}
	\ee
	We can show that $R_{T_1}(\w){\rightarrow}_p Q_0(\w)$ for any $\w\in\calH$ if $\{\blam_{t}\}$ and $\{\epsilon_{i,t}\}$ are both stationary for all $t$ such that $T_1^{-1}\sumT\beps_t\beps_t^\top{\rightarrow}_p \sigma^2_{\beps}\I_{J+1}$ and $T_1^{-1}\sumT\blam_t\blam_t^\top{\rightarrow}_p \bOme_0$. 
	This result implies that minimizing $R_{T_1}(\w)$ is asymptotically equivalent to minimizing $Q_0(\w)$, and thus our weight limit $\w^{\text{opt}}_{T_1}$ is asymptotically identical to the limit $\w^*_{\text{FP}}$ considered in \citet{ferman2021QE}  if the factors and idiosyncratic shocks are stationary.
	Due to this asymptotic equivalence, $\w_{T_1}^{\text{opt}}$ also fails to recover the factor structure. To see this more explicitly, note from~\eqref{R1 decomposition} that $R_{T_1}(\w)$ is  composed of two parts of errors when approximating the treated unit using the SC unit: the error of approximating the factors and the error of approximating the idiosyncratic shock, similar to the decomposition of $Q_0(\w)$. 
	As a result, $\w^{\text{opt}}_{T_1}$ needs to balance the two parts of errors in $R_{T_1}(\w)$, thus deviating from the minimizer of the first part and failing to recover the true factor structure, i.e., $(\bmu_0^\top,c_0)^\top\neq(\M^\top,\bc)^\top\w_{T_1}^{\text{opt}}$. It further implies that the resulting SC estimator generally does not converge to the targeted treatment effect $\alpha_{0,t}$, confirming the conclusion of \citet{ferman2021QE}. We complement \citet{ferman2021QE} by quantifying the rate of convergence. 
	
	
	The result in Theorem~\ref{th:con} also offers a verification of the important assumption (Assumption~3.2) of \cite{ferman2021} to guarantee the asymptotic unbiasedness of the SC estimator, i.e., there exists a weight vector $\w^*\in\mathcal{H}_{\text{orig}}$ such that $\left\|\w^{*}\right\|\to_p0$ and $\left\|\bmu_0-\M\w^{*}\right\|\to_p0$. 
	We show that as $J$ diverges, the optimal weight $\w^{\text{opt}}_{T_1}$ defined by~\eqref{w0} is a candidate choice that could satisfy Assumption 3.2 of \citet{ferman2021}, i.e., $\left\|\w^{\text{opt}}_{T_1}\right\|\to0$ and $\left\|\bmu_0-\M\w^{\text{opt}}_{T_1}\right\|\to0$.  
	To see this more explicitly, we focus on the weight $\w_{T_1}^\text{opt}\in\calH_{\text{orig}}$ and follow \citet{ferman2021} to assume that $\{\epsilon_{i,t}\}_{t\in\calT_0\cup\calT_1}$ are independent across $i$ and $\var(\epsilon_{i,t})=\sigma_{\epsilon}^2$ for the sake of simplification. We can (re-)define the factors and loadings, such that $c_i$ is absorbed into the factor structure $\blam_t^\top\bmu_i$ as \cite{ferman2021}.  Thus, $R_{T_1}(\w)$ can be written as
	\be
	R_{T_1}(\w)&=&\frac{1}{T_1}\sumT\left(\bmu_0-\M\w\right)^\top\blam_t\blam_t^\top\left(\bmu_0-\M\w\right)+\sigma^2_{\epsilon}\left(1+\w^\top\w\right).
	\label{R1 decomposition simplifed}
	\n
	\ee
	Denote $\Q=T_1^{-1}\sumT^\top\blam_t\blam_t^\top$. 
	We can analytically obtain the solution of $\w_{T_1}^\text{opt}$ by solving the optimization~\eqref{w0} with Karush-Kuhn-Tucker conditions as
	\be
	\w_{T_1}^\text{opt}=\left(\M^\top\Q\M+\sigma_{\epsilon}^2\I_{J}\right)^{-1}\left(\M^\top\Q\bmu_{0}+\frac{1}{2}\brho_1-\frac{1}{2}\rho_2\biota_{J}\right),
	\label{w_opt}
	\ee
	where $\brho_1$ is a nonnegative $J\times1$ constant vector,  $\rho_2$ is a constant, and $\biota_{J}$ is a $J\times 1$ vector of ones.  In the Online Appendix,
	we show that  $\left\|\bmu_0-\M\w^{\text{opt}}_{T_1}\right\|\to0$ and $\left\|\w^{\text{opt}}_{T_1}\right\|\to0$ hold if there exist positive constants $c_1$, $c_2$, $c_3$ and $c_4$ such that
	\be
	c_1 \leq \lambda_{min}\left(J^{-1}\M\M^\top\right) \leq  \lambda_{max}\left(J^{-1}\M\M^\top\right) \leq c_2
	\label{bound of MM}
	\ee
	and
	\be
	c_3 \leq \lambda_{min}\left(\Q\right) \leq  \lambda_{max}\left(\Q\right) \leq c_4.
	\label{bound of Q}
	\ee
Hence, our analysis of the convergence of SC weights provides conditions under which Assumption 3.2 of \citet{ferman2021} holds such that the asymptotic unbiasedness of the SC estimator can be achieved when $J$ diverges. Intuitively, it requires that the information contained in the factor structure neither diminishes nor diverges. Such upper and lower bounds guarantee that the factors are informative (to be able to reconstruct) and do not dominate the idiosyncratic shocks (so that the weights can dilute among control units), respectively.

	Theorem~\ref{th:con} also provides a new perspective for
        understanding the role of covariates. \cite{botosaru2019}
        shows that a perfect fit of observed covariates is not
        essential to achieve the asymptotic unbiasedness of SC
        estimators as long as the fit of outcomes is good, but a
        better fit of covariatesis associated with tighter
        bounds. Theorem~\ref{th:con} illustrates the role of
        covariates via the convergence rate, and it suggests that a
        better fit of covariates (hence a smaller $\xi_{T_0}$) helps promote the convergence of the SC weight. This result is in line with the study of bounds by \citet{botosaru2019}. 

	To conclude this subsection, the above discussion shows that the SC estimator constructed using the limiting optimal weight $\w_{T_1}^{\text{opt}}$ minimizes the expected MSPE but also suffers from an asymptotic bias under fixed $J$. This result suggests a bias-variance tradeoff when using the averaged outcome of the control units for treatment effect evaluation and motivates us to further study how the SC weight $\wh\w$ compare with other weighting schemes in terms of MSPE.
	
	\subsection{Asymptotic optimality of SC estimators} \label{sec:optimality}
	In this subsection, we investigate how the SC weights balance the bias and variance. We first establish the asymptotic optimality for SC estimators without an intercept and then consider the case with an intercept. We also compare the SC estimator with other treatment effect estimators that also involve a weighted average of control units.
	
	\subsubsection{Asymptotic optimality of SCM without an intercept}\label{sec:opt-no-intercept}
	We need some additional conditions.
	\begin{condition}
		\label{as:data mixing}
		For any $i\in\{0,1,\ldots,J\}$, $\left\{\epsilon_{i,t}\right\}$ is either $\alpha$-mixing with the mixing coefficient $\alpha=-r/(r-2)$ or $\phi$-mixing with the mixing coefficient $\phi=-r/(2r-1)$ for $r\geq 2$.
	\end{condition}
	Condition~\ref{as:data mixing} restricts the dependence of the idiosyncratic shocks. A similar assumption is needed in \cite{ferman2021} (see Assumption~3.1~(b)), while \cite{abadie2010} imposes a stronger requirement that $\left\{\epsilon_{i,t}\right\}$ are independent across units and over time.
	
	\begin{condition}\
		\begin{enumerate}[(i)]
			\item \label{as:moment of error term}
			There exists a constant $C_1$ such that $\E\epsilon^4_{i,t}\leq C_1<\infty$ for $i\in\{0,1,\ldots,J\}$ and $t\in\calT_0\cup\calT_1$.
			\item \label{as:var of pre error term}
			There exists a constant $C_2$ such that $\var\left(T_0^{-1/2}\sumt  e_{t,\epsilon}^{(i)}e_{t,\epsilon}^{(j)}\right)\geq C_2>0$ for all $T_0$ sufficiently large and any $i, j\in\{1,\ldots,J\}$.
			\item \label{as:var of post error term}
			There exists a constant $C_3$ such that  $\var\left(T_1^{-1/2}\sumT  e_{t,\epsilon}^{(i)}e_{t,\epsilon}^{(j)}\right)\geq C_3>0$ for all $T_1$ sufficiently large and any $i, j\in\{1,\ldots,J\}$.
		\end{enumerate}
		\label{as:error term}	
	\end{condition}
	Condition~\ref{as:error term} provides a set of regularity conditions to apply the central limit theorem for the dependent process; see \cite{schonfeld1971}, \cite{scott1973} 
	and \cite{wooldridge1988}. 
	Condition~\ref{as:error term}~\eqref{as:moment of error term} requires that all idiosyncratic shocks should not have heavy tails such that their fourth moments can be uniformly bounded. The same condition is also imposed in \citeauthor{xu2017} (2017, Assumption~4),  \citeauthor{botosaru2019} (2019, Assumption~3.1~(c)) and \citeauthor{ferman2021} (2021, Assumption~3.1~(c)). 
	Conditions~\ref{as:error term}~\eqref{as:var of pre error term}--\eqref{as:var of post error term} concern the difference between the idiosyncratic shock of the treated and control units. They guarantee that the variances of shocks do not degenerate as $T_0$ and $T_1$ increase, such that the asymptotic distributions can be properly defined. If the variances are degenerating, then the sample MS(P)E of the SC estimator $L_{T_0}(\w)$ and $L_{T_1}(\w)$ converge to their expectation $R_{T_0}(\w)$ and $R_{T_1}(\w)$, respectively, at a faster rate (see~\eqref{a0.1} and \eqref{a1.3} in the Appendix), which cannot be quantified by the standard central limit theorems, but the conclusion is then expected to hold more easily.

	\begin{condition}
		$\xi_{T_0}^{-1}T_0^{-1/2}J^2=o(1)$.
		\label{as:xi1}
	\end{condition}
	Condition~\ref{as:xi1} restricts the relative rate of several quantities going to infinity, i.e., $\xi_{T_0}$, $T_0$ and $J$. Importantly, note that this condition implies that $\xi_{T_0}\neq0$, which turns out to be a crucial condition to establish the asymptotic optimality of the SC weight. Intuitively, $\xi_{T_0}\neq0$ means that it is not possible to perfectly fit the pretreatment outcomes and observed covariates of the treated unit using a linear combination of the covariates and outcomes of the control units, and we refer to this situation as imperfect pretreatment fit. Imperfect pretreatment fit can result from multiple sources. To see this, note that $\xi_{T_0}=\inf_{\w \in \calH} R_{T_0}(\w)$, and we can decompose $R_{T_0}(\w)$ as
		\be
		R_{T_0}(\w)	&=&\left(\wt\bmu_{0}-\sumj w_j\wt\bmu_{j}\right)^\top \frac{1}{T_0}\sumt \wt\blam_{t}\wt\blam_{t}^\top \left(\wt\bmu_{0}-\sumj w_j\wt\bmu_{j}\right) \n\\
		&&+\frac{1}{T_0}\E\left\| \Z_0-\sumj w_j\Z_j \right\|^2+\frac{1}{T_0}\sumt\E\left(\epsilon_{0,t}-\sumj w_j\epsilon_{j,t}\right)^2,\n
		\ee
		where $\wt\blam_{t}=(\delta_t,1,\bthe_t^\top,\blam_{t}^\top)^\top$ for $t\in\calT_0\cup\calT_1$ and $\wt\bmu_i=(1,c_i,\Z_i^\top,\bmu_{i}^\top)^\top$ for $i\in\{0,\ldots,J\}$. A poor fit of any of the three components in $R_{T_0}(\w)$ can lead to an imperfect pretreatment fit. For example, an irrelevant donor pool can lead to a poor reconstruction of the factor structure and covariates, which further deteriorates the fit. A poor pretreatment fit can also be caused by sizeable idiosyncratic shocks since the last part of $R_{T_0}(\w)$ is substantial when the variance of shocks is large. Moreover, imperfect fit also occurs when a weight vector does not simultaneously kill the approximation errors of factor structure, covariates and shocks.

	We also discuss how our definition of imperfect pretreatment fit is related to those in the literature. This concept was first formally presented by \cite{abadie2010}, in which they define a perfect pretreatment fit as the existence of a weight $\w\in\calH_{\text{orig}}$ that satisfies $y_{0,t}=\sum_{j=1}^Jw_j y_{j,t}$ for all $t\in\mathcal{T}_0$ and $\Z_0=\sum_{j=1}^Jw_j\Z_j$, implying that $\xi_{T_0}=0$. The same definition is also used by \citet{botosaru2019} and \citet{eli2021}, among others. Thus our definition is in line with these studies.
	\cite{ferman2021QE} defines ``imperfect pre-treatment fit'' as  the (possible) nonexistence of $\w^\ast\in\calH_{\text{orig}}$ that satisfies $y_{0,t}=\sum_{j=1}^J w^\ast_j y_{jt}$ for every $t\in\calT_0$, implying that $\xi_{T_0}$ can be non-zero. Hence, our definition is also compatible with theirs.

%

\begin{theorem} \label{th:opt}
	If $T_1$ is finite, then under Conditions~\ref{as:data stochastic}--\ref{as:post error term bias}, \ref{as:data mixing},
	\ref{as:error term}~\eqref{as:moment of error term}--\eqref{as:var of pre error term} and \ref{as:xi1}, we have 	
	\begin{equation}
		\frac{ R_{T_1}(\wh \w) }
		{\inf_{\w \in \calH} { R_{T_1}(\w)} } \overset{p}{\rightarrow}1.
		\label{opt0}
	\end{equation}
	If $T_1$ diverges at rate $O(T_0)$, then under Conditions~\ref{as:data stochastic}--\ref{as:post error term bias} and \ref{as:data mixing}--\ref{as:xi1}, we have \eqref{opt0} and
	\begin{equation}
		\frac{ L_{T_1}(\wh \w) }
		{\inf_{\w \in \calH} { L_{T_1}(\w)} } \overset{p}{\rightarrow}1;
		\label{opt1}
	\end{equation}
	furthermore, if $\xi_{T_1}^{-1}\left\{L_{T_1}(\wh\w)-\xi_{T_1}\right\}$ is uniformly integrable, then
	\begin{equation}
		\frac{\E L_{T_1}(\wh \w) }
		{\inf_{\w \in \calH} { R_{T_1}(\w)} } \rightarrow 1.
		\label{opt2}
	\end{equation}
\end{theorem}
Theorem~\ref{th:opt} establishes the asymptotic optimality of the SC estimator, and the form differs depending on whether $T_1$ is divergent and the randomness of weights are incorporated. Specifically, when $T_1$ is finite, 	
\eqref{opt0} shows that the SC weight is asymptotically optimal among all possible weighting schemes in the sense that the risk of the SC estimator $R_{T_1}(\wh \w)$ is asymptotically identical to that of the infeasible best estimator. When $T_1$ goes to infinity at the same rate as $T_0$, we can state a similar optimality but in terms of the squared error $L_{T_1}(\w)$, a sample counterpart of $R_{T_1}(\w)$, as in \eqref{opt1}. Further examination of the proof reveals that \eqref{opt0} is one of the sufficient conditions of \eqref{opt1}. Finally, note that $R_{T_1}(\wh\w)$ and $L_{T_1}(\wh\w)$ are both obtained by replacing the unknown weight with the estimated SC weight $\wh\w$, i.e., $R_{T_1}(\wh\w)=R_{T_1}(\w)\big|_{\w=\wh\w}$ and $L_{T_1}(\wh\w)=L_{T_1}(\w)\big|_{\w=\wh\w}$, 
and thus neither \eqref{opt0} nor \eqref{opt1} accounts for the randomness of $\wh\w$. Therefore, \eqref{opt2} establishes the asymptotic optimality regarding $\E L_{T_1}(\wh\w)$, where the expectation is taken with respect to $y_{i,t}^N$ and $\wh\w$, such that the randomness in $\wh\w$ is explicitly incorporated. Overall, the result in~\eqref{opt2} shows that the expected squared error of the SC estimator, accounting for the randomness of the SC weight, is asymptotically identical to the minimum risk achieved by the infeasible best weight among all possible weighting schemes in the set $\mathcal{H}$. Theorem~\ref{th:opt} also holds in the absence of $\Z$, implying that balancing between pretreatment outcomes and covariates is not essential as long as $\xi_{T_0}\neq 0$.

The asymptotic optimality of the SC weight is inspired by optimal
model averaging, once we observe the link between the SC estimator and
the model averaging estimator. Optimal model averaging concerns the
bias-variance tradeoff in the presence of model uncertainty and is
intended to obtain the best prediction by optimally combining
estimators obtained from candidate models with different
specifications. While asymptotic optimality is one of the most
important properties in optimal model averaging studies, almost all
works focus on the risk assuming that the weights and data are
fixed. The only exception is \cite{zhang2021}, which analyzes the risk
while accounting for the randomness of weights, but they only consider in-sample risk. Since we average the \emph{posttreatment} outcomes of control units, which are treated as random and can be regarded as an out-of-sample extension of the pretreatment outcomes, we need to incorporate the randomness of data and weights and study out-of-sample prediction risk. Thus, we contribute to the model averaging literature by providing the first out-of-sample asymptotic optimality accounting for the randomness of data and weights. 
Moreover, we also generalize existing optimality analysis in the model averaging by allowing for negative weights. See \citet{radchenko2021} for detailed discussions on the effect of negative weights in the combination.

In practice, there are other treatment effect estimators that construct the counterfactual outcome also using a weighted average of the control units. Two popular examples include the matching estimator and IPW.
Specifically, the matching estimator constructs the counterfactual outcome of the treated unit based on a set of matched units \citep[see, e.g.,][]{rosenbaum1983,dehejia2002,abadie2006}. In the case of only one treated unit, denote the fixed constant $K$ as the number of matches, and let $\calJ_K$ be the set of matches for the treated unit (determined based on, e.g., covariates or propensity scores). Then, the matching counterfactual estimator can be written as
$\wh y_{0,t}^N(\wh\w^{\text{match}})=\sum_{j=1}^J \wh w^{\text{match}}_{j} y_{j,t}^N$, where  
\be\label{eq:matching}
\wh w^{\text{match}}_j=\left\{
\begin{array}{ll}
	1/K & \text{if}\ j\in\calJ_K(0)\\
	0 & \text{otherwise}
\end{array}
\right..
\n
\ee
Obviously, $\wh\w^{\text{match}}$ also belongs to $\calH$.

The normalized IPW method constructs the weights based on the propensity score \citep[see, e.g.,][]{imbens2004}. 
Denote $\calS$ as the set of indices of treated units and $\mathbb{I}(\cdot)$ as an indicator function. In our framework, $\calS=\{0\}$ and $\mathbb{I}(i\in\calS)=1$ only when $i=0$.
For some prespecified characteristics $\X_i$ for unit $i$, e.g., $\X_i=\left(y_{i,-T_0+1},\ldots,y_{i,0},\Z_i^\top\right)^\top$ in our setup, the propensity score of unit $i$ is defined as $\pi(\X_i)=\Pr\left\{ \mathbb{I}(i\in\calA) \mid \X_i\right\}$, that is, the conditional probability of unit $i$ being treated.
Then, the IPW estimator can be written as $\wh y_{0,t}^N(\wh\w^{\text{IPW}})=\sum_{j=1}^J \wh w^{\text{IPW}}_j y_{j,t}^N$, where 
\be\label{eq:ipw}
\wh w^{\text{IPW}}_j&=&\left\{\sum_{i=1}^J\frac{(1-d_{i,t})\pi\left(\X_i\right)}{1-\pi\left(\X_i\right)}\right\}^{-1}{ \frac{(1-d_{j,t})\pi\left(\X_j\right)}{1-\pi\left(\X_j\right)} }\n\\
&=&\left\{\sum_{i=1}^J\frac{\pi\left(\X_i\right)}{1-\pi\left(\X_i\right)}\right\}^{-1}{ \frac{\pi\left(\X_j\right)}{1-\pi\left(\X_j\right)} },
\ee
and this weight also satisfies $\wh\w^{\text{IPW}}\in\calH$. 
From Theorem~\ref{th:opt}, we know that the SC weight is potentially
more desirable than other weights if one's aim is to achieve the minimum MSPE. This implies that matching and IPW estimators cannot outperform the SC estimator in terms of achieving the lowest risk, at least in the asymptotic sense. 

Theorem~\ref{th:opt} provides another justification for the SC
estimator, especially under imperfect pretreatment fit and a finite
number of control units. Specifically, \citet{ferman2021QE} shows that
the SC estimator is biased under imperfect pretreatment fit, and
\citet{ferman2021} shows that this bias disappears only when the
number of control units goes to infinity (see also the discussion in Section~\ref{sec:convergence}). Our result in Theorem~\ref{th:opt} suggests that although the SC estimator is biased, it is asymptotically optimal in terms of achieving the minimum (expected) squared prediction error among all possible estimators that construct counterfactual outcomes based on averaging control units. Such asymptotic optimality holds regardless of whether the number of control units $J$ is finite. Thus, our results significantly widen the range of applicability of the SC estimator, showing that it is still a recommended method even under imperfect pretreatment fit with a finite number of control units. 
Note further that the MSPE can be decomposed into the variance and the square of bias. Thus, in the case where the SC estimator is unbiased due to diverging $J$ (as shown by \cite{ferman2021} and Theorem~\ref{th:con}), the asymptotic optimality in Theorem~\ref{th:opt} implies that the variance of the SC estimator converges to the lower bound. 
Related to the regret analysis in \citet{chen2022}, our asymptotic optimality also implies the bound convergence of the corresponding regret (with a potentially different rate from \citet{chen2022}) in some situations. This is because the denominators in~\eqref{opt0}--\eqref{opt2}, i.e., the infimum of the (expected) MSPEs, are typically of a constant (or even lower) order under certain regularity conditions, and thus the convergence of the ratios of (expected) MSPEs in~\eqref{opt0}--\eqref{opt2} implies that the bound of the corresponding regret, e.g., $R_{T_1}(\wh\w)-\inf_{\w \in \calH} R_{T_1}(\w)$, converges. 
Besides, Theorem~\ref{th:opt} also provides theoretical foundation for the numerical finding of \citet{bottmer2021} that the root mean squared errors of SC estimators are substantially lower than difference-in-means in their simulation studies, but they do not theoretically prove this.

\subsubsection{Asymptotic optimality of SCM with an intercept}
\label{sec:optimality with intercept}
\citet{doudchenko2016} argues that many of the treatment effect estimators in the literature employ a common linear structure to construct the counterfactual outcome of the treated unit, i.e.,
\be
\wh y_{0,t}^N(\w,d)=\sum_{j=1}^J w_j y_{j,t}^N+d,
\label{estimation with intercept}
\ee
where $\w=(w_1,\ldots,w_J)^\top\in\calH$ and $d$ is an intercept whose parameter space is $\Ra$. In this framework, the standard SC estimator is to set $d=0$ and chooses the optimal weight within $\calH_{\text{orig}}$. An alternative popular estimator is DID, which relaxes the restriction of $d=0$ but imposes that all weights are equal across control units \citep{athey2006,doudchenko2016}, i.e., 
$\wh w^\text{DID}_{j}=1/J$ for $j\in\{1,\ldots,J\}$ and $$\wh d^{\text{DID}}=\frac{1}{T_0}\sumt y_{0,t}^N-\frac{1}{T_0J}\sumt\sumj y_{j,t}^N.$$
To enjoy the flexibility of non-zero intercepts in DID but still
maintain the advantage of data-driven weights in the SCM, \citet{doudchenko2016} and \citet{ferman2021QE} propose a demeaned SC (DSC) estimator, given by
$$
\wh\alpha^{\text{DSC}}_{0,t}=y_{0,t}-\sumj \wh w_j^{\text{DSC}} y_{j,t}-\wh d^{\text{DSC}},
$$
where $\wh w^{\textrm{DSC}}_j$ and $\wh d^{\text{DSC}}$ are the demeaned version of the SC weight and intercept and can be obtained by 
\be
(\wh w^{\textrm{DSC}}_j,\wh d^{\text{DSC}})=\underset{\w\in\calH,\ d\in\Ra}{\arg\min}L_{T_0}(\w,d),
\label{SC weight with intercept}
\ee
where 
\be
L_{T_0}(\w,d)=\frac{1}{T_0}
\left\{\sum_{t\in\calT_0}\left(y_{0,t}-\sumj w_j y_{j,t}-d\right)^2+\left\|\Z_0-\sumj w_j\Z_j-d\biota_{r}\right\|^2\right\},
\n
\ee
with $\biota_{r}$ being an $r\times 1$ vector of ones. Note that the
resulting intercept estimate from~\eqref{SC weight with intercept}
coincides the estimated intercept of \citet{ferman2021QE} in standard
cases.\footnote{To illustrate the relation between the intercept
  estimated by~\eqref{SC weight with intercept} and that of
  \citet{ferman2021QE}, we assume that $\Ra=[D_\text{L},D_\text{U}]$ with $D_\text{L}$ and $D_\text{U}$ being some constants. \cite{ferman2021QE} calculate $\wh w_j^\text{DSC}$ by
$\underset{\w\in\calH}{\min}1/T_0
\sum_{t\in\calT_0}\left[y_{0,t}-\sumj w_j y_{j,t} - \left(y^-_0-\sumj w_j y^-_j \right) \right]^2$, with $\bar y^-_i=T_0^{-1}\sumt y_{i,t}$. In the absence of covariates, this is equivalent to solving~\eqref{SC weight with intercept} subject to $d=y^-_0-\sumj w_j y^-_j$ as long as $D_\text{L}<\wh d^\text{DSC} < D_\text{U}$. Since it is not difficult to define $D_\text{L}$ and $D_\text{U}$ such that $\wh d^\text{DSC}$ is an interior point of $\mathcal{D}$, the equivalence holds in most standard cases.
}

\citet{ferman2021QE} shows that incorporating the intercept relaxes
the condition for the SC estimator to be unbiased. Specifically, when
the treatment assignment is not correlated with time-varying
unobservables, the DSC estimator is asymptotically unbiased like DID,
despite that its weights still fail to recover the time-invariant fixed effects of the treated unit. The DSC estimator also leads to a lower MSPE than DID. In contrast, when the treatment assignment does correlate with time-varying unobservables with $\E\blam_t\neq0$ for $t\in\mathcal{T}_1$, neither the DSC nor the DID estimator is asymptotically unbiased. \citet{ferman2021QE} also claims that in general it is impossible to rank these two estimators in terms of bias and MSPE in this case. 

In this section, we investigate the asymptotic optimality of the DSC estimator. Some additional conditions are needed. 
    
\begin{condition}
	There exists a sufficiently large constant $C_d$ such that $|d|\leq C_d$ for any $d\in\Ra\subset\mathbb{R}$.
	\label{as:intercept bound}
\end{condition}
This condition restricts the range of extrapolation caused by allowing
for the intercept term. Since $C_d$ can be sufficiently large, this
restriction is mild. Additionally, note that this boundedness
requirement can be satisfied under a compact parameter space, which is commonly used in the econometric literature. 

\begin{condition}\ 
	\begin{enumerate}[(i)]
		\item \label{as:var of pre error term with intercept}
		There exists a constant $\wt C_2$ such that $\var\left\{T_0^{-1/2}\sumt  e_{t,\epsilon}^{(j)}\right\}\geq\wt C_2>0$ for all $T_0$ sufficiently large and any $ j\in\{1,\ldots,J\}$.
		\item \label{as:var of post error term with intercept}
		There exists a constant $\wt C_3$ such that  $\var\left\{T_1^{-1/2}\sumT  e_{t,\epsilon}^{(j)}\right\}\geq\wt C_3>0$ for all $T_1$ sufficiently large and any $j\in\{1,\ldots,J\}$.
	\end{enumerate}
	\label{as:error term with intercept}	
\end{condition}
This condition resembles Conditions~\ref{as:error term}~\eqref{as:var of pre error term}--\eqref{as:var of post error term}, needed 
for the central limit theorem for dependent processes. It guarantees that the variance of the (difference in) idiosyncratic shocks is nondegenerate as $T_0$ and $T_1$ increase.
\begin{condition}\
	\label{as:factor bound with intercept}
	\begin{enumerate}[(i)]
		\item \label{lambda with intercept}	${T_1}^{-1}\sumT\left({T_0}^{-1}\sum_{k\in\calT_0}\blam_k-\blam_t\right)=O\left(T_0^{-1/2}\right).$
		\item \label{delta with intercept}  ${T_1}^{-1}\sumT\left({T_0}^{-1}\sum_{k\in\calT_0}\delta_k-\delta_t\right)=O\left(T_0^{-1/2}\right).$
		\item \label{theta with intercept}	${T_1}^{-1}\sumT\left({T_0}^{-1}\sum_{k\in\calT_0}\bthe_k-\bthe_t\right)=O\left(T_0^{-1/2}\right).$
	\end{enumerate}
\end{condition}
Condition~\ref{as:factor bound with intercept} requires that the factor loadings, time fixed effects and the slope coefficients do not change substantially after the treatment. It plays a similar role as Conditions~\ref{as:factor bound} and~\ref{as:covariates}~\eqref{theta}. Note that the ``diverging'' factors in the example discussed below Condition~\ref{as:factor bound} also satisfy Condition~\ref{as:factor bound with intercept}.

Let $R_{T_0}(\w,d)=\E L_{T_0}(\w,d)$ and $\wt\xi_{T_0}=\inf_{\w \in \calH,d\in\Ra}R_{T_0}(\w,d)$.
\begin{condition}
	\label{as:post error term bias with intercept}
	$\wt\xi_{T_0}^{-1}\supw\left| {T_0}^{-1}\sumt\E\left(\sumj w_j  e_{t,\epsilon}^{(j)}\right)^2-{T_1}^{-1}\sumT\E\left(\sumj w_j  e_{t,\epsilon}^{(j)}\right)^2\right|=o(1)$.
\end{condition}

\begin{condition}
	$\wt\xi_{T_0}^{-1}T_0^{-1/2}J^2=o(1)$.
	\label{as:xi1 with intercept}
\end{condition}
Conditions~\ref{as:post error term bias with intercept} and~\ref{as:xi1 with intercept} are demeaned version of Conditions~\ref{as:post error term bias} and~\ref{as:xi1}, respectively. Note that $\wt\xi_{T_0}=\inf_{\w \in \calH,d\in\Ra}R_{T_0}(\w,d)\leq\inf_{\w \in \calH} R_{T_0}(\w,0)=\xi_{T_0}$,
and thus Conditions~\ref{as:post error term bias with intercept} and~\ref{as:xi1 with intercept} are both slightly stronger than their  counterparts without demeaning. Nevertheless, most standard setups of the SCM would still satisfy these two conditions, including the example described below Condition~\ref{as:post error term bias}. 

\begin{theorem} \label{th:opt with intercept}
	If $T_1$ is finite and Conditions~\ref{as:data stochastic}--\ref{as:covariates}, \ref{as:data mixing},
	\ref{as:error term}~\eqref{as:moment of error term}--\eqref{as:var of pre error term}, \ref{as:intercept bound}, \ref{as:error term with intercept}~\eqref{as:var of pre error term with intercept} and \ref{as:factor bound with intercept}--\ref{as:xi1 with intercept} hold, then 	
	\begin{equation}
		\frac{ R_{T_1}(\wh\w^{\emph{DSC}},\wh d^{\emph{DSC}}) }
		{\inf_{\w \in \calH,d\in\Ra} { R_{T_1}(\w,d)} } \overset{p}{\rightarrow}1.
		\label{opt0 with intercept}
	\end{equation}
	If $T_1$ diverges at rate $O(T_0)$ and
	Conditions~\ref{as:data stochastic}--\ref{as:covariates}, \ref{as:data mixing}--\ref{as:error term} and \ref{as:intercept bound}--\ref{as:xi1 with intercept} hold, then we have \eqref{opt0 with intercept} and
	\begin{equation}
		\frac{ L_{T_1}(\wh \w^{\emph{DSC}},\wh d^{\emph{DSC}}) }
		{\inf_{\w \in \calH,d\in\Ra} { L_{T_1}(\w,d)} } \overset{p}{\rightarrow}1;
		\label{opt1 with intercept}
	\end{equation}
	furthermore, if $\wt\xi_{T_1}^{-1}\left\{L_{T_1}(\wh\w^{\emph{DSC}},\wh d^{\emph{DSC}})-\wt\xi_{T_1}\right\}$ is uniformly integrable, then
	\begin{equation}
		\frac{\E L_{T_1}(\wh \w^{\emph{DSC}},\wh d^{\emph{DSC}}) }
		{\inf_{\w \in \calH,d\in\Ra} { R_{T_1}(\w,d)} } \rightarrow 1.
		\label{opt2 with intercept}
	\end{equation}
\end{theorem}
Theorem~\ref{th:opt with intercept} establishes the asymptotic optimality of the DSC estimator.
It shows that the DSC estimator can still achieve the minimum
(expected) squared prediction risk (and loss) asymptotically even if
it is biased when the treatment assignment is correlated with
time-varying unobservables. Moreover, it also suggests that the DSC
estimator is not worse than the DID estimator in terms of the squared
prediction error, at least asymptotically. Note that this result does
not conflict with the statement of \citet{ferman2021QE} that it is in general impossible to compare the bias and MSPE between the DSC and the DID estimators. Our optimality is in the asymptotics, and we focus on the squared error, and the finite sample bias and MSPE comparison between these two estimators remains unclear.

Intuitively, the asymptotic optimality of the SC estimator (with or without an intercept) is a consequence of the fact that the objective function to estimate the SC weight is consistent with the evaluation criterion, namely the (expected) squared error, while other estimators, such as matching, IPW and DID, construct weights using a different objective function or simply impose equal weights. 
For this reason, the SCM can also be viewed as a task-based approach, whose advantages in prediction are illustrated in \citep{donti2017}.
The property of asymptotic optimality is also precisely aligned with the goal of the SCM---to synthesize a good control unit to predict the counterfactual of the treated---and it provides the conditions under which the SCM may outperform other estimators. 


\subsection{Asymptotic properties in a model-free setup}
\label{sec:general-model}
While the linear factor model covers a wide range of random processes to generate potential outcomes, one may still be reluctant to completely disregard other DGPs. Yet, the applicability of SCM does not seem to be restrained by certain specific forms of models. To reconcile these two perspectives, we 
investigate the asymptotic behavior of SCM without specifying a linear factor model as the outcome process. We show that the convergence result and the asymptotic optimality of SCM continue to hold in a model-free setup. By imposing assumptions directly on outcomes (rather than factor structures), some arguments even simplify. To save space, we only provide the assumptions and (re)state the general versions of Theorems~\ref{th:con} and~\ref{th:opt} here but relegate the detailed discussions, the case with intercepts, and proofs to the Online Appendix. 

First, we provide conditions needed to examine the convergence of SC weights in a model-free steup.
\begin{condition}
	\label{as:data stochastic relax dgp}
	We treat $\left\{ \Z_i\mid i\in\{0,1,\ldots,J\} \right\}$ as fixed and $\left\{ y^N_{i,t}\mid i\in\{0,1,\ldots,J\}, t\in\calT_0\cup\calT_1 \right\}$ as stochastic.
\end{condition}

\begin{condition}
	\label{as:risk bound} 
	$\supw\left| T_0^{-1}\sumt\E\left(y^N_{0,t}-\sumj w_j y^N_{j,t}\right)^2-R_{T_1}(\w) \right|=O(T_0^{-1/2}J^2)+o(\xi_{T_0}).$
\end{condition}

Condition~\ref{as:data stochastic relax dgp} is a general version of Condition~\ref{as:data stochastic}~\eqref{as:data fixed}. Condition~\ref{as:risk bound} restricts the difference between the fits in the pre- and posttreatment periods. Its direct implication is that the main difference between the pre- and posttreatment outcomes is exclusively due to the treatment effect. This condition can be viewed as a generalization of Conditions~\ref{as:factor bound}--\ref{as:post error term bias}. 

\begin{theorem} \label{th:con relax dgp}	
	Given any $T_1$, if $\w^{\text{opt}}_{T_1}$ is an interior point of $\calH$ and Conditions~\ref{as:covariates}~\eqref{Z}, \ref{as:Y_c} and \ref{as:data stochastic relax dgp}--\ref{as:risk bound} hold, then
	$\left\|\wh\w-\w^{\text{opt}}_{T_1}\right\|=O_p\left( T_0^{\nu}\xi_{T_0}^{1/2}+T_0^{\nu}\xi_{T_1}^{1/2}+T_0^{-1/4+\nu}J\right)$, where $\nu>0$ is a sufficiently small constant.
\end{theorem}
This theorem is a restatement of the convergence of SC weights (Theorem~\ref{th:con}) without assuming a linear factor model as the DGP. 

Next, we provide additional conditions needed for the asymptotic optimality in a model-free setup. 
\begin{condition}
	\label{as:data mixing relax dgp}
	For any $i\in\{0,1,\ldots,J\}$, $\left\{y^N_{i,t}\right\}$ is either $\alpha$-mixing with the mixing coefficient $\alpha=-r/(r-2)$ or $\phi$-mixing with the mixing coefficient $\phi=-r/(2r-1)$ for $r\geq 2$.
\end{condition}
Denote $e_{t,y^N}^{(i)}=y^N_{0,t}-y^N_{i,t}$ for $i\in\{1,\ldots,J\}$ and $t\in\calT_0\cup\calT_1$.
\begin{condition}\
	\begin{enumerate}[(i)]
		\item \label{as:moment of error term relax dgp}
		There exists a constant $C_4$ such that $\E \left(y^{N}_{i,t}\right)^4\leq C_4<\infty$ for $i\in\{0,1,\ldots,J\}$ and $t\in\calT_0\cup\calT_1$.
		\item \label{as:var of pre error term relax dgp}
		There exists a constant $C_5$ such that $\var\left(T_0^{-1/2}\sumt  e_{t,y^N}^{(i)}e_{t,y^N}^{(j)}\right)\geq C_5>0$ for all $T_0$ sufficiently large and any $i, j\in\{1,\ldots,J\}$.
		\item \label{as:var of post error term relax dgp}
		There exists a constant $C_6$ such that  $\var\left(T_1^{-1/2}\sumT  e_{t,y^N}^{(i)}e_{t,y^N}^{(j)}\right)\geq C_6>0$ for all $T_1$ sufficiently large and any $i, j\in\{1,\ldots,J\}$.
	\end{enumerate}
	\label{as:error term relax dgp}	
\end{condition}

Conditions~\ref{as:data mixing relax dgp} and~\ref{as:error term relax dgp} generalize the factor-model-based Conditions~\ref{as:data mixing} and~\ref{as:error term}, respectively.

\begin{theorem} \label{th:opt relax dgp}
	If $T_1$ is finite, then under Conditions~\ref{as:covariates}~\eqref{Z}, \ref{as:xi1}, \ref{as:data stochastic relax dgp}--\ref{as:data mixing relax dgp} and
	\ref{as:error term relax dgp}~\eqref{as:moment of error term relax dgp}--\eqref{as:var of pre error term relax dgp}, we have~\eqref{opt0}; 
	If $T_1$ diverges at rate $O(T_0)$, then under Conditions~\ref{as:covariates}~\eqref{Z}, \ref{as:xi1} and \ref{as:data stochastic relax dgp}--\ref{as:error term relax dgp}, we have \eqref{opt0} and \eqref{opt1}; 
	Furthermore, if $\xi_{T_1}^{-1}\left\{L_{T_1}(\wh\w)-\xi_{T_1}\right\}$ is uniformly integrable, then we have~\eqref{opt2}. 
\end{theorem}
This theorem is a generalized version of Theorem~\ref{th:opt}, relaxing the linear factor model assumption.
\section{Simulation}\label{sec:simulation}
In this section, we verify the theory via simulation. We first examine the convergence of the SC weight and then compare the SC estimators with popular competing methods to verify their asymptotic optimality.

\subsection{Simulation design}\label{sec:design}
We follow \cite{hsiao2019} to generate the data from the following pure factor model:
\be
y_{i,t}^N=\gamma_{1,i}f_{1,t}+\gamma_{2,i}f_{2,t}+u_{i,t},  
\n
\ee
where the common factors $f_{s,t}$ and the factor loadings $\gamma_{s,i}$, $s \in\{1,2\}$ are both drawn independently from $\N(0,1)$. The idiosyncratic shock $u_{i,t}$ is weakly cross-sectionally dependent, generated by 
\be
u_{i,t}&=&(1+b^2)v_{i,t}+bv_{i+1,t}+bv_{i-1,t},\n
\ee
where $v_{i,t}\sim \text{i.i.d.} \N(0,\sigma_i^2)$, $b=1$ and $\sigma_i^2$ are drawn independently from $0.5\left(\chi^2(1)+1\right)$ for all $i$.
As above, we set unit 0 as the only treated unit, and the remaining $\{1,\ldots,J\}$ units are the control.
Since the value of $\alpha_{i,t}$ does not influence the estimation
and evaluation procedure based on \eqref{SC weight} and \eqref{LT1},
we follow the literature on the in SCM not assigning $\alpha_{i,t}$ a value.
We set $J\in\{30, 50\}$, the number of pretreatment periods $T_0\in\{50, 100, 200, 400\}$ and the number of posttreatment periods $T_1=10$. The number of replications is $R=1000$.

\subsection{Simulation results on weight convergence}\label{sec:result}
To investigate the convergence of the SC weight, we 
need to know the limit of the SC weight, which further requires the knowledge of $R_{T_1}(\w)$. If we denote $\bet_{t}=\left(y_{0,t}^N,\ldots,y_{J,t}^N\right)^\top$, then $\bet_{t}$ follows $\N\left(\bmu_{\bet_t},\bSig_{\bet_t}\right)$, where $\bmu_{\bet_t}$ is a $(J+1)\times1$ vector with its $(i+1)$-th element being $\mu_{\bet_t,i+1}=\gamma_{1,i}f_{1,t}+\gamma_{2,i}f_{2,t}$ and $\bSig_{\bet_t}$ is a $(J+1)\times(J+1)$ matrix with its element in the $(i+1)$-th row and $(j+1)$-th column being
$$
\bSig_{\bet_t,i+1,j+1}=\left\{
\begin{array}{ll}
	\left(1+b^2\right)^2\sigma_0^2+b^2\sigma_1^2 & if\ i=j=0,\\
	\left(1+b^2\right)^2\sigma_i^2+b^2\left(\sigma_{i-1}^2+\sigma_{i+1}^2\right) & if\ 1\leq i=j\leq J-1,\\
	\left(1+b^2\right)^2\sigma_J^2+b^2\sigma_{J-1}^2 & if\ i=j=J,\\
	b\left(1+b^2\right)\left(\sigma_i^2+\sigma_j^2\right) & if\ |i-j|=1,\\
	b^2\sigma_{(i+j)/2}^2 & if\ |i-j|=2,\\
	0 & \text{otherwise}.
\end{array}
\right.
$$ 
Note further that $\left\{\bet_{t}\mid t\in\mathcal{T}_0\cup\mathcal{T}_1\right\}$ are independent over $t$. 
With $\bmu_{\bet_t}$ and $\bSig_{\bet_t}$ at hand, one can compute $R_{T_1}(\w)$ as
\be
R_{T_1}(\w)&=&\E L_{T_1}(\w)  \n\\
&=&\frac{1}{T_1}\sumT\E\left(y_{0,t}^N-\sumj w_j y_{j,t}^N \right)^2  \n\\
&=&\frac{1}{T_1}\sumT\E\left\{ \bet_t^\top\left(1,-\w^\top\right)^\top\left(1,-\w^\top\right)\bet_t\right\}  \n\\
&=&\frac{1}{T_1}\sumT\bmu_{\bet_t}^\top\left(1,-\w^\top\right)^\top\left(1,-\w^\top\right)\bmu_{\bet_t}+\tr\left\{\left(1,-\w^\top\right)^\top\left(1,-\w^\top\right)\bSig_{\bet_t}\right\}.
\label{simu:R}
\ee
When an intercept is allowed for, the generalized version of $R_{T_1}(\w,d)$ can be written as
\be
R_{T_1}(\w,d)
&=&\E L_{T_1}(\w,d)=\frac{1}{T_1}\sumT\E\left(y_{0,t}^N-\sumj w_j y_{j,t}^N-d \right)^2\n\\
&=&\frac{1}{T_1}\sumT\left(\bmu_{\bet_t}^\top,1\right)\left(1,-\w^\top,-d\right)^\top\left(1,-\w^\top,-d\right)\left(\bmu_{\bet_t}^\top,1\right)^\top
\n\\
&&+\tr\left\{\left(1,-\w^\top\right)^\top\left(1,-\w^\top\right)\bSig_{\bet_t}\right\}.
\n
\ee

In the simulation, we search for the weight in the original weight set $\mathcal{H}_{\text{orig}}$, i.e., setting $C_{\text{L}}=0$ and $C_{\text{U}}=1$. The SC weight is estimated by  $$\wh\w=\underset{\w\in\mathcal{H}_{\text{orig}}}{\arg\min} \frac{1}{T_0}\sum_{t\in\calT_0}\left(y_{0,t}-\sumj w_j y_{j,t}\right)^2,$$ and its limit is obtained by $\w^{\text{opt}}_{T_1}=\underset{\w\in\calH_{\text{orig}}}{\arg\min} R_{T_1}(\w)$ with $R_{T_1}(\w)$ computed from~\eqref{simu:R}. Allowing for negative weights does not qualitatively change the results.

Figure~\ref{fig:weight} plots the vector norm of the difference
between the SC weight and its limit,
$\left\|\wh\w-\w^{\text{opt}}_{T_1}\right\|$, averaged over the 1000
replications, as $T_{0}$ increases. The solid line indicates that $J=30$, and the dashed line is when $J=50$. 
Under both sample sizes, $\left\|\wh\w-\w^{\text{opt}}_{T_1}\right\|$ is monotonically decreasing as $T_0$ increases, which confirms the convergence result in Theorem~\ref{th:con}. Comparing the values obtained under different numbers of control units, we find that $\wh\w$ converges faster when $J=30$ than $J=50$, which again confirms that the rate of convergence slows when $J$ increases as stated in Theorem~\ref{th:con}.

\begin{figure}[htbp]
	\centering
	\includegraphics[scale=0.45]{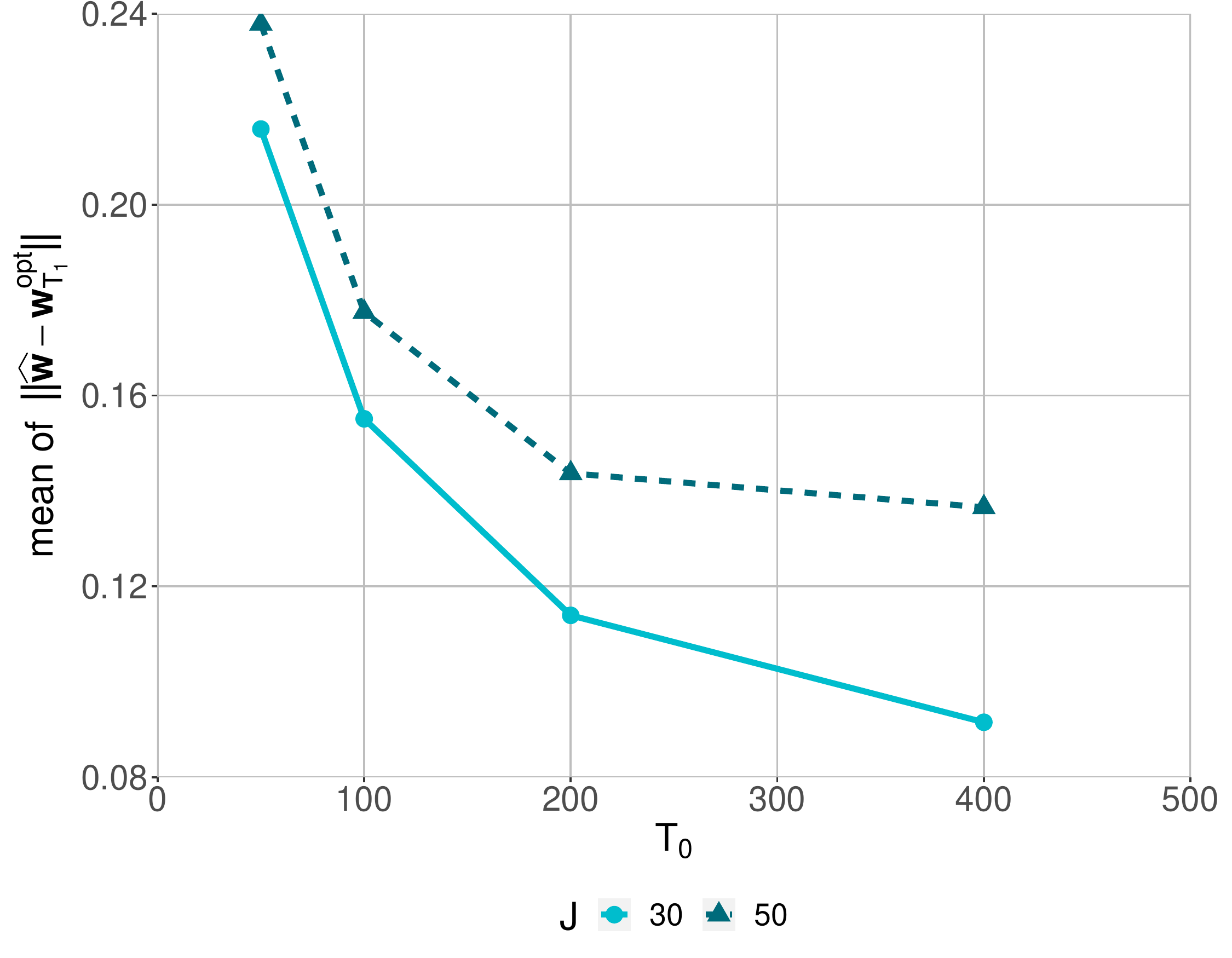}
	\caption{Average of $\left\|\wh\w-\w^{\text{opt}}_{T_1}\right\|$ over 1000 replications}
	\label{fig:weight}
\end{figure}

\subsection{Simulation results on asymptotic optimality}
To examine the asymptotic optimality of the SC estimator, we compare it with alternative estimators that also construct counterfactual outcomes by averaging control units. We consider two versions of SCM, the standard SCM and the SCM with an intercept (denoted as DSC). We compare them with five alternative methods, i.e., propensity score matching (PSM), IPW, equal-weight averaging (Equal), best control selection (Sel) and DID. 

The standard SCM and DSC are implemented as \eqref{SC weight} and \eqref{SC weight with intercept}, respectively.
PSM is one of the most popular matching methods, first proposed by \cite{rosenbaum1983}. It is based on the propensity score, $\wh\pi(\X_i)$ estimated from a logistic regression with $\X_i$ as regressors. 
The weights are computed by~\eqref{eq:matching}, where we set $K=1$ and 
$$
\calJ_K(0)=\left\{j\in\{1,\ldots,J\} \mid \sum_{k=1}^J \mathbb{I}\Big(\left| \wh\pi\left(\X_k\right)-\wh\pi\left(\X_0\right)\right| \leq \left| \wh\pi\left(\X_j\right)-\wh\pi\left(\X_0\right)\right|\Big)\leq K \right\},
$$
with $\mathbb{I}(\cdot)$ being an indicator function.  
IPW computes the weight from~\eqref{eq:ipw}, also using the estimated propensity score $\wh\pi(\X_i)$ as above.
Equal-weight averaging constructs the counterfactuals by using simple average of all control units. The best control selection selects a single control unit based on minimizing the in-sample mean squared error. 
It can be viewed as a special case of best subset selection; see \cite{doudchenko2016} for further explanation. Furthermore, it can also 
be viewed as a special case of averaging controls, since the weight
equals 1 for the selected best control and zeros for the remaining controls, i.e.,
$$
\wh w^{\text{Sel}}_j=\left\{
\begin{array}{ll}
	1 & \text{if}\ j={\arg\min}_{ m_0\in\{1,\ldots,J\}}{T_0}^{-1}\sumt\left(y_{0,t}^N-y_{m_0,t}^N\right)^2,\\
	0 & \text{otherwise}.
\end{array}
\right.
$$
We evaluate all methods by $R_{T_1}(\w)$ computed from~\eqref{simu:R}.

Figure~\ref{fig:Rwd} plots the ratio of risk,
$R_{T_1}(\w)/\inf_{\w^{*} \in \calH}R_{T_1}(\w^{*})$ for the methods
without intercepts and $R_{T_1}(\w,d)/\inf_{\w^{*} \in \calH,
  d^{*}\in\Ra}R_{T_1}(\w^{*},d^{*})$ for the methods with intercepts,
averaged over replications. SCM and DSC generally perform quite
similarly, and both methods exhibit clear superiority as their ratios
of the risk are much lower than those of the other methods for all
$T_0$ and $J$. Moreover, the curves of SCM and DSC both monotonically
decrease toward 1 as $T_0$ increases, implying that the risk of the
SCM and DSC estimators converges to the lowest possible risk as the
number of pretreatment periods increases. This result precisely
coincides with the asymptotic optimality stated in Theorems~\ref{th:opt} and~\ref{th:opt with intercept}. 
Regarding the competing methods, we note that PSM and Sel produce a
very similar ratio of risk because they typically select the same
control unit to construct the counterfactual. The performance of IPW,
Equal and DID are also similar. The similarity between IPW and Equal
can be explained by the fact that our DGP generates all $J$ control
units from the same distribution, resulting in
$\wh\w^{\text{IPW}}\approx 1/J$, and the similarity between Equal and
DID is also expected because both impose equal weights with the only
difference being in the intercept. Figure~\ref{fig:Rwd} shows that the
ratios of all competing methods do not converge, implying that none of
these methods achieves asymptotic optimality in terms of risk.

\begin{figure}[htbp]
	\centering
	\subfigure{
		\includegraphics[scale=0.33]{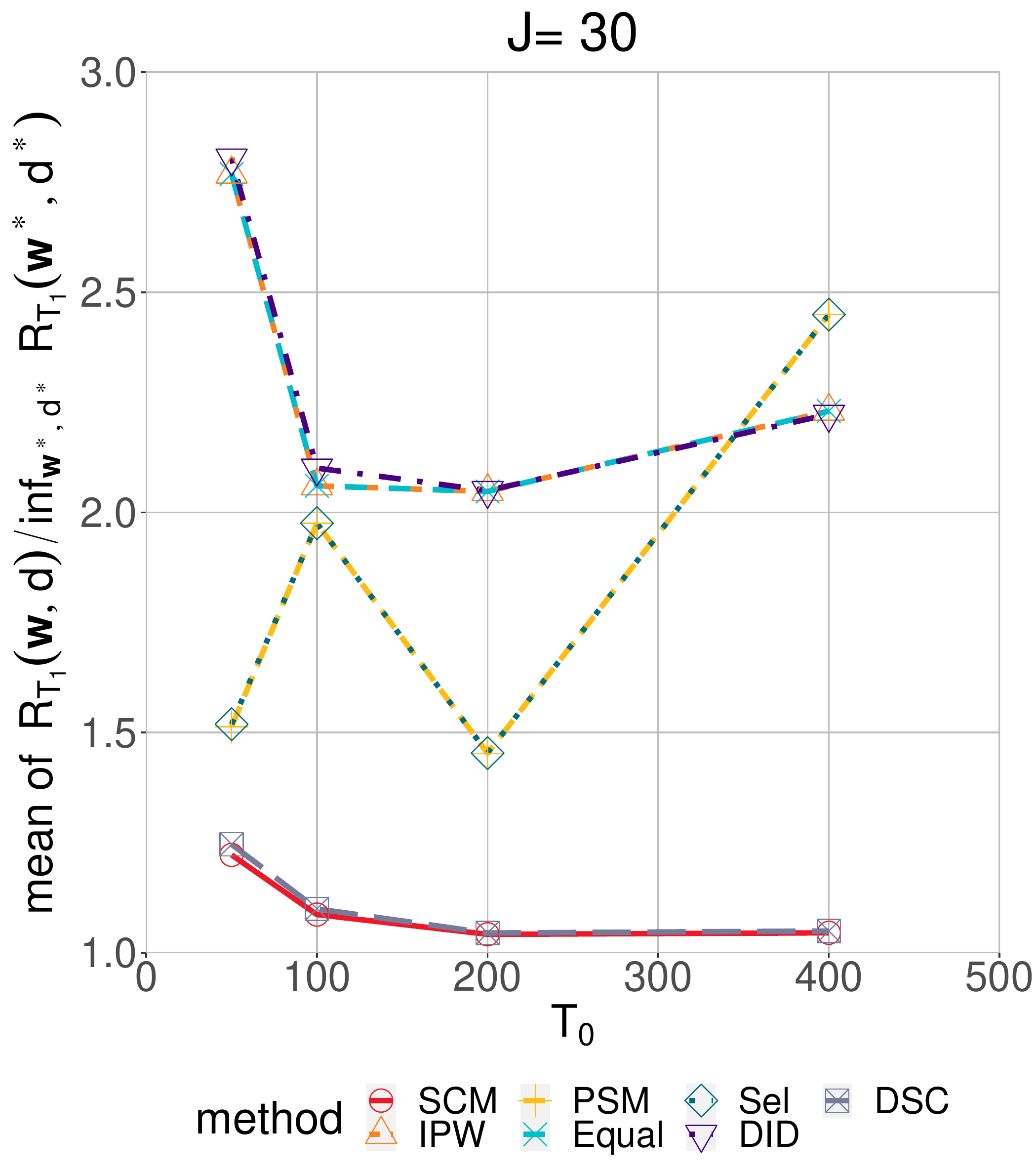}
	}
	\subfigure{
		\includegraphics[scale=0.33]{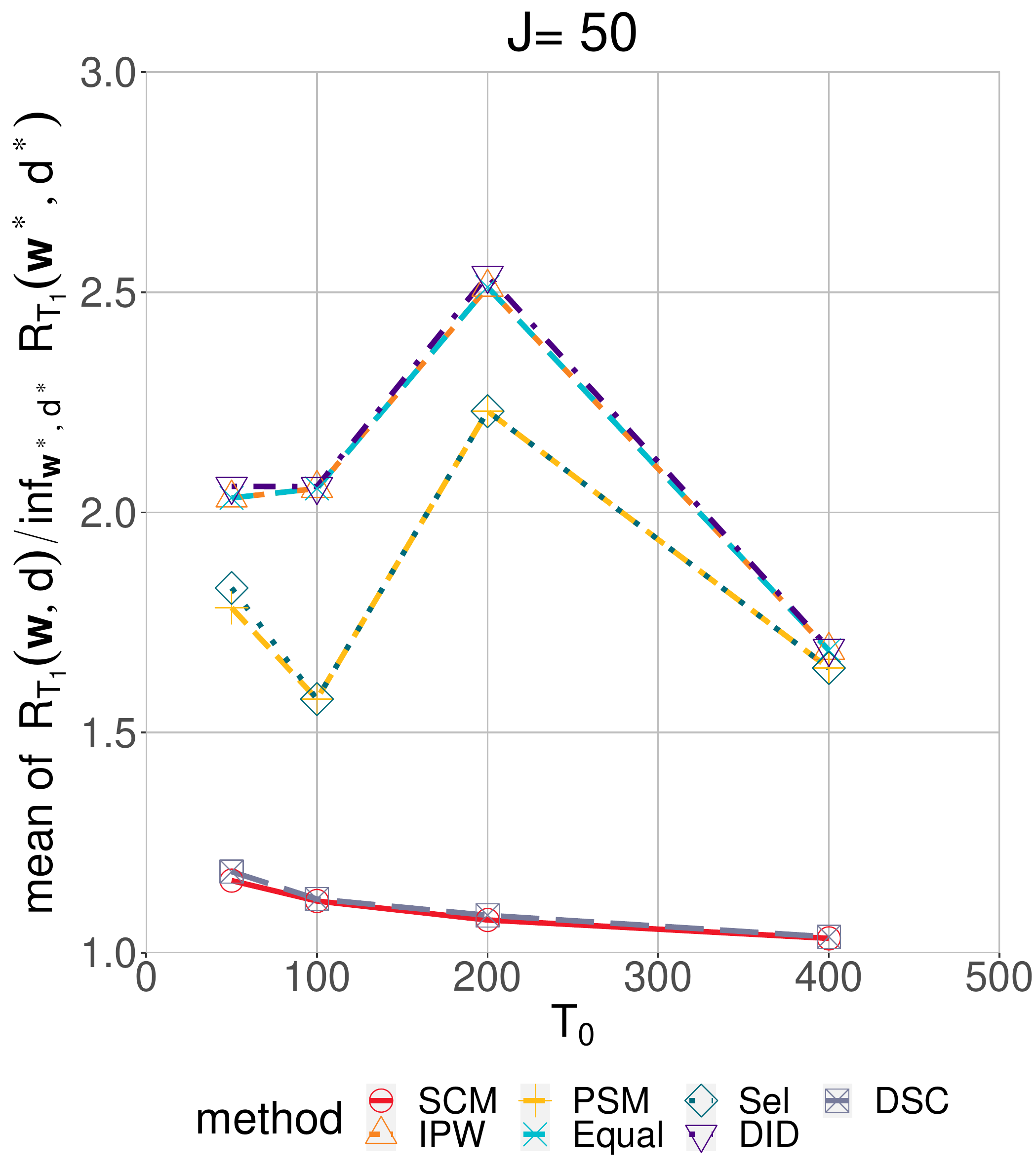}
	}
	\caption{Average of $R_{T_1}(\w,d)/\inf_{\w^* \in \calH,d^*\in\Ra}R_{T_1}(\w^*,d^*)$ over 1000 replications}
	\label{fig:Rwd}
\end{figure}

\section{Conclusion}\label{sec:discussion}
This paper investigates the asymptotic properties of the SCM as the number of the pretreatment periods diverges. 
We show that the SC weight converges to the limit that may not recover the factor structure but minimizes the (expected) mean squared prediction error. 
We quantify the rate of convergence, from
which one can better understand how the number of control units and
the pre- and posttreatment fit influence the convergence rate. Our convergence results also verify under which conditions the SC weight can dilute and reconstruct the factor structure as the number of controls increases, so that the unbiasedness of the SC estimator can be achieved. 
	
Furthermore, we establish the asymptotic optimality of both the
standard SCM and DSC estimators. We provide conditions under which the
two versions of SC estimators are asymptotically optimal among the
class of estimators that construct counterfactual outcomes using an
average of control units. Our asymptotic optimality suggests that if
there is no weight to produce a perfect pretreatment fit, then the
(expected) squared prediction error of the SC estimator converges to
the lowest possible risk, implying that the SC estimator asymptotically dominates other estimators in this class, such as matching, IPW and DID. This result justifies the use of the SC estimator even when there are unobserved confounders, pretreatment fit is not perfect, and the number of control units is finite. These properties are also free of the underlying assumptions on the DGP of potential outcomes.

Our theoretical analysis suggests that the SCM provides theoretical
guarantees and can be used in a wide
range of applications, regardless of whether pretreatment fit is perfect. When pretreatment fit is perfect, the SCM estimator can be interpreted as an unbiased estimator of the treatment effect; in the presence of imperfect pretreatment fit, the SCM can still be applied as an asymptotically minimum-MSPE estimator of the treatment effect. 

\baselineskip=16pt
\bibliographystyle{abbrvnat}
\bibliography{reference_Properties_of_SCM}


\setcounter{equation}{0} \renewcommand{\theequation}{\thesection.\arabic{equation}}

\newpage \appendix
\singlespace
\setcounter{table}{0} \renewcommand{\thetable}{\thesection\arabic{table}} %
\setcounter{figure}{0} \renewcommand{\thefigure}{\thesection\arabic{figure}}

\section*{APPENDIX}
\section{Proof of Theorem~\ref{th:con}}
\label{sec:pf2}
Denote $\tau_{T_0}= T_0^{\nu}\xi_{T_0}^{1/2}+T_0^{\nu}\xi_{T_1}^{1/2}+T_0^{-1/4+\nu}J$, 
and let $\u=(u_1,\ldots,u_J)^\top$.  According to \cite{fan2004} and \cite{lu2015}, to prove Theorem~\ref{th:con}, it suffices to show that, for any $\varepsilon>0$, there exists a constant $C_\varepsilon$, such that 
	\be
	\Pr\left\{
	\inf_{\left\|\u\right\|=C_\varepsilon,\left(\w^{\text{opt}}_{T_1}+\tau_{T_0}\u\right)\in\calH}
	L_{T_0}(\w^{\text{opt}}_{T_1}+\tau_{T_0}\u)\geq L_{T_0}(\w^{\text{opt}}_{T_1})
	\right\}>1-\varepsilon,
	\label{a2.1}
	\ee
	under a given $T_1$ and any sufficient large $T_0$.
\eqref{a2.1} implies that there exists a $\w^{\star}$ in the bounded closed domain set $\left\{\w^{\text{opt}}_{T_1}+\tau_{T_0}\u \mid  \left(\w^{\text{opt}}_{T_1}+\tau_{T_0}\u\right)\in\calH, \left\|\u\right\|\leq C_\varepsilon\right\}$, such that it (locally) minimizes $L_{T_0}(\w)$ and $\left\| \w^{\star} - \w^{\text{opt}}_{T_1} \right\|=O_p(\tau_{T_0})$. From the convexity of $L_{T_0}(\w)$ and $\calH$, $\w^{\star}$ is also the unique global minimizer, i.e., $\w^{\star}=\wh\w$.

Define $D(\u)=L_{T_0}(\w^{\text{opt}}_{T_1}+\tau_{T_0}\u)-L_{T_0}(\w^{\text{opt}}_{T_1})$. Then, we can decompose $D(\u)$ as
\be
D(\u)
&=&\frac{1}{T_0} \left\| \Y_0-{\boldsymbol{\mathcal{Y}}_c}(\w^{\text{opt}}_{T_1}+\tau_{T_0}\u) \right\|^2-\frac{1}{T_0}\left\| \Y_0-{\boldsymbol{\mathcal{Y}}_c}\w^{\text{opt}}_{T_1} \right\|^2
\n\\
&=&-\frac{{2}\tau_{T_0}}{T_0} \left(\Y_0-{\boldsymbol{\mathcal{Y}}_c}\w^{\text{opt}}_{T_1}\right)^\top{\boldsymbol{\mathcal{Y}}_c}\u+\frac{\tau_{T_0}^2}{T_0} \left\|{\boldsymbol{\mathcal{Y}}_c}\u\right\|^2
\n\\
&\equiv&\Delta_{1}+\Delta_{2},
\n
\ee
where $\Delta_{1}=-{2}\tau_{T_0}T_0^{-1} \left(\Y_0-{\boldsymbol{\mathcal{Y}}_c}\w^{\text{opt}}_{T_1}\right)^\top{\boldsymbol{\mathcal{Y}}_c}\u$ and $\Delta_{2}=\tau_{T_0}^2T_0^{-1} \left\|{\boldsymbol{\mathcal{Y}}_c}\u\right\|^2$. We show that $\Delta_{2}$ is the dominant term of $D(\u)$ as follows.

We first consider $\Delta_{2}$. From Condition~\ref{as:Y_c}, we have that, with probability approaching 1,
\be
\kappa_1\leq \lambda_{min}\left(\frac{1}{T_0}{\boldsymbol{\mathcal{Y}}_c}^\top{\boldsymbol{\mathcal{Y}}_c}\right)\leq\lambda_{max}\left(\frac{1}{T_0}{\boldsymbol{\mathcal{Y}}_c}^\top{\boldsymbol{\mathcal{Y}}_c}\right)\leq\kappa_2.
\label{a2.2}
\ee
This further implies that, with probability approaching 1,
\be
\Delta_{2}
\geq
\frac{\tau_{T_0}^2}{T_0}\lambda_{min}\left({\boldsymbol{\mathcal{Y}}_c}^\top{\boldsymbol{\mathcal{Y}}_c}\right)\left\|\u\right\|^2
\geq
\kappa_1\tau_{T_0}^2 \left\|\u\right\|^2. 
\label{a2.3}
\ee

Next, we consider $\Delta_{1}$. To simplify the notation, denote $\wt\blam_{t}=(\delta_t,1,\bthe_t^\top,\blam_{t}^\top)^\top$ for $t\in\calT_0\cup\calT_1$ and $\wt\bmu_i=(1,c_i,\Z_i^\top,\bmu_{i}^\top)^\top$ for $i\in\{0,\ldots,J\}$. 
Let $\e_{\wt\bmu}^{(i)}=\wt\bmu_{0}-\wt\bmu_{i}$ for $i\in\{0,\ldots,J\}$.
Under our linear factor structure~\eqref{DGP-factor model} and Condition~\ref{as:data stochastic}, we have that
\be
&&\supw\left| R_{T_0}(\w)-R_{T_1}(\w) \right|
\n\\
&=&\supw\left| 
\frac{1}{T_0}\sumt\left(\sumj w_j\wt\blam_t^\top \e_{\wt\bmu}^{(j)} \right)^2+\frac{1}{T_0}\sumt\E\left(\sumj w_j  e_{t,\epsilon}^{(j)}\right)^2+\frac{1}{T_0}\left\|\Z_0-\sumj w_j\Z_j\right\|^2\right.
\n\\
&&\left.-\frac{1}{T_1}\sumT\left(\sumj w_j\wt\blam_t^\top \e_{\wt\bmu}^{(j)} \right)^2-\frac{1}{T_1}\sumT\E\left(\sumj w_j  e_{t,\epsilon}^{(j)}\right)^2
\right|
\n\\
&\leq&\supw\left| 
\frac{1}{T_0}\sumt\left(\sumj w_j\wt\blam_t^\top \e_{\wt\bmu}^{(j)} \right)^2-\frac{1}{T_1}\sumT\left(\sumj w_j\wt\blam_t^\top \e_{\wt\bmu}^{(j)} \right)^2\right|
+\frac{1}{T_0}\supw\left\|\Z_0-\sumj w_j\Z_j\right\|^2
\n\\
&&+\supw\left| 
\frac{1}{T_0}\sumt\E\left(\sumj w_j  e_{t,\epsilon}^{(j)}\right)^2-\frac{1}{T_1}\sumT\E\left(\sumj w_j  e_{t,\epsilon}^{(j)}\right)^2\right|
\n\\
&\equiv& {\mathbb{I}}_1+{\mathbb{I}}_2+{\mathbb{I}}_3.
\label{a2.4}
\ee 
We analyze ${\mathbb{I}}_1$, ${\mathbb{I}}_2$ and ${\mathbb{I}}_3$ in turn. First, we analyze ${\mathbb{I}}_1$. From Conditions~\ref{as:factor bound}--\ref{as:covariates}, we have 
\be
\frac{1}{T_1}\sumT\left(\frac{1}{T_0}\sum_{k\in\calT_0}\wt\blam_k^\top\wt\blam_k-\wt\blam_t^\top\wt\blam_t\right)=O(T_0^{-1/2}),
\n
\ee
and the components of $\wt\bmu_i$ are bounded; hence,
\be
{\mathbb{I}}_1
&=&\supw\left| \sumi\sumj w_i w_j
\left\{ \frac{1}{T_0}\sumt\wt\blam_t^\top \e_{\wt\bmu}^{(i)}\e_{\wt\bmu}^{(j)\top}\wt\blam_t-\frac{1}{T_1}\sumT\wt\blam_t^\top \e_{\wt\bmu}^{(i)}\e_{\wt\bmu}^{(j)\top}\wt\blam_t
\right\} \right|
\n\\
&\leq&\supw \sumi\sumj |w_i| |w_j|\left|
\frac{1}{T_0}\sumt\wt\blam_t^\top \e_{\wt\bmu}^{(i)}\e_{\wt\bmu}^{(j)\top}\wt\blam_t-\frac{1}{T_1}\sumT\wt\blam_t^\top \e_{\wt\bmu}^{(i)}\e_{\wt\bmu}^{(j)\top}\wt\blam_t
\right|
\n\\
&\leq&\left[\max\{C_\text{L}, C_\text{U}\}\right]^2\sumi\sumj \left|
\frac{1}{T_0}\sumt\wt\blam_t^\top \e_{\wt\bmu}^{(i)}\e_{\wt\bmu}^{(j)\top}\wt\blam_t-\frac{1}{T_1}\sumT\wt\blam_t^\top \e_{\wt\bmu}^{(i)}\e_{\wt\bmu}^{(j)\top}\wt\blam_t
\right|
\n\\
&=&\left[\max\{C_\text{L}, C_\text{U}\}\right]^2 \sumi\sumj \left| \tr\left(
\frac{1}{T_0}\sumt \e_{\wt\bmu}^{(i)}\e_{\wt\bmu}^{(j)\top}\wt\blam_t\wt\blam_t^\top\right)-\tr\left(
\frac{1}{T_1}\sumT \e_{\wt\bmu}^{(i)}\e_{\wt\bmu}^{(j)\top}\wt\blam_t\wt\blam_t^\top\right)
\right|
\n\\
&=&\left[\max\{C_\text{L}, C_\text{U}\}\right]^2 \sumi\sumj \left| \tr\left\{\left(
\frac{1}{T_0}\sumt\wt\blam_t\wt\blam_t^\top-\frac{1}{T_1}\sumT \wt\blam_t\wt\blam_t^\top\right)\e_{\wt\bmu}^{(i)}\e_{\wt\bmu}^{(j)\top}\right\}
\right|
\n\\
&=&O(T_0^{-1/2}J^2).
\label{a2.5}
\ee
Here, the last equality is satisfied due to the fixed $r$ and $F$. We then examine $\mathbb{I}_2$. From Condition~\ref{as:covariates}~\eqref{Z}, we have that 
\be
\mathbb{I}_2
&\leq&\frac{1}{T_0}\supw\left\{\sumj\left\|w_j\left(\Z_0-\Z_j\right)\right\|\right\}^2
\n\\
&\leq&\left[\max\{C_\text{L}, C_\text{U}\}\right]^2T_0^{-1}\left\{\sumj\left\|\Z_0-\Z_j\right\|\right\}^2
\n\\
&=&O_p(T_0^{-1}J^2).
\label{a2.6}
\ee 
Finally, we consider $\mathbb{I}_3$. From 
Condition~\ref{as:post error term bias}, it is obvious that $\mathbb{I}_3=o(\xi_{T_0})$. Combining with Inequalities~\eqref{a2.4}--\eqref{a2.6}, it implies that
\be
\supw\left| R_{T_0}(\w)-R_{T_1}(\w) \right|=O(T_0^{-1/2}J^2)+o(\xi_{T_0}).
\label{a2.7}
\ee
The above equation holds for $\w^{\text{opt}}_{T_1}$, and note that $R_{T_1}(\w^{\text{opt}}_{T_1})=\xi_{T_1}$. Thus, we have that
	\be
	R_{T_0}(\w^{\text{opt}}_{T_1})-\xi_{T_1} =O(T_0^{-1/2}J^2)+o(\xi_{T_0}).
	\label{a2.8}
	\ee
	Since $\E\left\|\Y_0-{\boldsymbol{\mathcal{Y}}_c}\w^{\text{opt}}_{T_1}\right\|^2\leq T_0R_{T_0}(\w^{\text{opt}}_{T_1})$, using \eqref{a2.8}, we can obtain that
\be
\left\|\Y_0-{\boldsymbol{\mathcal{Y}}_c}\w^{\text{opt}}_{T_1}\right\|=O_p(T_0^{1/2}\xi_{T_1}^{1/2})+O_p(T_0^{1/4}J)+o_p(T_0^{1/2}\xi_{T_0}^{1/2}).
\label{a2.9}
\ee
Therefore, we have that
\be
\left|\Delta_{1}\right|&\leq&
\frac{{2}\tau_{T_0}}{T_0}\left\|\Y_0-{\boldsymbol{\mathcal{Y}}_c}\w^{\text{opt}}_{T_1}\right\|
\left\|{\boldsymbol{\mathcal{Y}}_c}\u\right\|
\n\\
&\leq&\frac{{2}\tau_{T_0}}{T_0}\left\|\Y_0-{\boldsymbol{\mathcal{Y}}_c}\w^{\text{opt}}_{T_1}\right\|\sqrt{\lambda_{max}\left({\boldsymbol{\mathcal{Y}}_c}^\top{\boldsymbol{\mathcal{Y}}_c}\right)} \left\|\u\right\|
\n\\
&=&
O_p(\tau_{T_0}\xi_{T_1}^{1/2})\left\|\u\right\|+O_p(\tau_{T_0}T_0^{-1/4}J) \left\|\u\right\|+o_p(\tau_{T_0}\xi_{T_0}^{1/2})\left\|\u\right\|,\n
\ee
where the second inequality is due to~\eqref{a2.2} and the last equality is due to~\eqref{a2.9}.
This equation together with~\eqref{a2.3} shows that $\Delta_{2}$ asymptotically dominates $\Delta_{1}$. 
Therefore, $D(\u)\geq 0$ in probability for any $\u$ that satisfies $\left\|\u\right\|=C_\varepsilon$ and $\left(\w^{\text{opt}}_{T_1}+\tau_{T_0}\u\right)\in\calH$.
This completes the proof of~\eqref{a2.1}, and thus Theorem~\ref{th:con} holds.

\section{Proof of Theorem~\ref{th:opt}}
\label{sec:pf1}
\setcounter{equation}{0} 

We first prove \eqref{opt0} in Theorem~\ref{th:opt}. To this end, we decompose $L_{T_0}(\w)$ as
	\be
	L_{T_0}(\w)=R_{T_1}(\w)+L_{T_0}(\w)-R_{T_0}(\w)+R_{T_0}(\w)-R_{T_1}(\w).
	\n
	\ee 
	By Lemma~1 in \cite{gao2019}, it suffices to show that  
	\be
	\supw \left| \frac{L_{T_0}(\w)-R_{T_0}(\w)}{ R_{T_0}(\w)} \right|=o_p(1)
	\label{a0.1}
	\ee
	and
	\be
	\supw \left| \frac{R_{T_0}(\w)-R_{T_1}(\w)}{ R_{T_0}(\w)} \right|=o(1).
	\label{a0.2}
	\ee
	
	We first verify~\eqref{a0.1}. Note that
	\be
	&&\supw \left| \frac{L_{T_0}(\w)-R_{T_0}(\w)}{ R_{T_0}(\w)} \right|
	\n\\
	&\leq&
	\xi_{T_0}^{-1}\supw \left| L_{T_0}(\w)-R_{T_0}(\w) \right|
	\n\\
	&=&
	\xi_{T_0}^{-1}\supw\left| \frac{1}{T_0}\sumt \left\{\left(y_{0,t}-\sumj w_j y_{j,t}\right)^2-\E\left(y_{0,t}-\sumj w_j y_{j,t}\right)^2 \right\}\right|
	\n\\
	&=&
	\xi_{T_0}^{-1}\supw\left| \frac{1}{T_0}\sumt\sumi\sumj w_i w_j \left\{ \left(y_{0,t}-y_{i,t}\right)\left(y_{0,t}-y_{j,t}\right)-\E\left(y_{0,t}-y_{i,t}\right)\left(y_{0,t}-y_{j,t}\right) \right\} \right|
	\n\\
	&\leq&
	\xi_{T_0}^{-1}\supw\sumi\sumj |w_i| |w_j|\left| \frac{1}{T_0}\sumt \left\{ \left(y_{0,t}-y_{i,t}\right)\left(y_{0,t}-y_{j,t}\right)-\E\left(y_{0,t}-y_{i,t}\right)\left(y_{0,t}-y_{j,t}\right) \right\}
	\right|
	\n\\
	&\leq&
	\left[\max\{C_\text{L}, C_\text{U}\}\right]^2\xi_{T_0}^{-1}\sumi\sumj \left| \frac{1}{T_0}\sumt \left\{ \left(y_{0,t}-y_{i,t}\right)\left(y_{0,t}-y_{j,t}\right)-\E\left(y_{0,t}-y_{i,t}\right)\left(y_{0,t}-y_{j,t}\right) \right\}
	\right|
	\n\\
	&=&
	\left[\max\{C_\text{L}, C_\text{U}\}\right]^2\xi_{T_0}^{-1}T_0^{-1/2}\sumi\sumj\Psi_{T_0}(i,j),
	\label{a0.1.1}
	\ee
	where 
	\be
	\Psi_{T_0}(i,j)=\left| \frac{1}{\sqrt{T_0}}\sumt \left\{ \left(y_{0,t}-y_{i,t}\right)\left(y_{0,t}-y_{j,t}\right)-\E\left(y_{0,t}-y_{i,t}\right)\left(y_{0,t}-y_{j,t}\right) \right\}
	\right|.
	\n
	\ee
	Under the linear factor structure~\eqref{DGP-factor model} and Condition~\ref{as:data stochastic}, we can rewrite $\Psi_{T_0}(i,j)$ as a function of idiosyncratic shocks as 
		\be
		\Psi_{T_0}(i,j)=\left| \frac{1}{\sqrt{T_0}}\sumt \left\{  e_{t,\epsilon}^{(i)}e_{t,\epsilon}^{(j)}-\E  e_{t,\epsilon}^{(i)}e_{t,\epsilon}^{(j)} \right\}
		\right|,
		\n
		\ee
	where we recall that $e_{t,\epsilon}^{(i)}=\epsilon_{0,t}-\epsilon_{i,t}$.
	From Condition~\ref{as:data mixing} and Theorem~3.49 in \cite{white1984}, $e_{t,\epsilon}^{(i)}e_{t,\epsilon}^{(j)}$ for $i$, $j\in\{1,\ldots,J\}$ is either an $\alpha$-mixing sequence with the mixing coefficient $\alpha=-r/(r-2)$ or a $\phi$-mixing sequence with the mixing coefficient $\phi=-r/(2r-1)$, $r\geq 2$. Moreover,  $\var\left(e_{t,\epsilon}^{(i)}e_{t,\epsilon}^{(j)}\right)$ can be uniformly bounded, as a result of Condition~\ref{as:error term}~\eqref{as:moment of error term}.
	Using Theorem~5.20 in \cite{white1984}, such properties of $e_{t,\epsilon}^{(i)}e_{t,\epsilon}^{(j)}$ together with Condition~\ref{as:error term}~\eqref{as:var of pre error term} imply that
	\be
	\sumi\sumj \Psi_{T_0}(i,j)=O_p(J^2).
	\label{a0.1.2}
	\ee
	Combining~\eqref{a0.1.1}, \eqref{a0.1.2} and Condition~\ref{as:xi1}, we can obtain~\eqref{a0.1}. 
	
	Then, we show that~\eqref{a0.2} holds.
Condition~\ref{as:data stochastic} implies~\eqref{a2.4}, and Conditions~\ref{as:factor bound}--\ref{as:covariates} imply Inequalities~\eqref{a2.5} and \eqref{a2.6}. With~\eqref{a2.4}--\eqref{a2.6} and  Condition~\ref{as:post error term bias}, we have that
	\be
	&&\supw\left| \frac{R_{T_0}(\w)-R_{T_1}(\w)}{ R_{T_0}(\w)} \right|
	\n\\
	&\leq&\xi_{T_0}^{-1}\supw\left| R_{T_0}(\w)-R_{T_1}(\w) \right|
	\n\\
	&=&O(\xi_{T_0}^{-1}T_0^{-1/2}J^2)+o(1).\n
	\ee
	This equation, combined with Condition~\ref{as:xi1}, leads to~\eqref{a0.2}.

	Note that the proof above holds regardless of whether $T_1$ is finite or divergent, as does~\eqref{opt0} in Theorem~\ref{th:opt}.

Next, we prove~\eqref{opt1} in Theorem~\ref{th:opt} when $T_1$ diverges at rate $O(T_0)$.
Since 
\be
L_{T_0}(\w)=L_{T_1}(\w)+L_{T_0}(\w)-R_{T_0}(\w)+R_{T_0}(\w)-R_{T_1}(\w)+R_{T_1}(\w)-L_{T_1}(\w)
\n
\ee
and \eqref{a0.1}--\eqref{a0.2} hold when $T_1$ diverges at rate $O(T_0)$,
by Lemma~1 in \cite{gao2019}, it suffices to show that for a divergent $T_1$ at rate $O(T_0)$,
\be
\supw \left| \frac{L_{T_1}(\w)-R_{T_1}(\w)}{R_{T_0}(\w)} \right|=o_p(1).
\label{a1.3}
\ee
Similar to the arguments of~\eqref{a0.1.1}--\eqref{a0.1.2}, under Conditions~\ref{as:data stochastic}, \ref{as:data mixing}, \ref{as:error term}~\eqref{as:moment of error term} and \ref{as:error term}~\eqref{as:var of post error term} and a divergent $T_1$ at rate $O(T_0)$, we can obtain that
\be
\supw \left| \frac{L_{T_1}(\w)-R_{T_1}(\w)}{ R_{T_0}(\w)} \right|=O_p(\xi_{T_0}^{-1}T_0^{-1/2}J^2).
\n
\ee
This equation together with Condition~\ref{as:xi1} leads to~\eqref{a1.3}, which further verifies~\eqref{opt1}.

Finally, we prove \eqref{opt2} in Theorem~\ref{th:opt} when $T_1$ is divergent at rate $O(T_0)$. 
We can bound $\xi_{T_1}^{-1}L_{T_1}(\wh\w)$ from both above and below as 
\be
\label{a1.4}
&&\xi_{T_1}^{-1}L_{T_1}(\wh\w)
\n\\
&=&\supw\left| \frac{L_{T_1}(\wh\w)}{L_{T_1}(\w)}\frac{L_{T_1}(\w)}{R_{T_0}(\w)} \frac{R_{T_0}(\w)}{R_{T_1}(\w)} \right|
\n\\
&\leq&\supw\left| \frac{L_{T_1}(\wh\w)}{L_{T_1}(\w)}\right|
\supw\left|\frac{L_{T_1}(\w)}{ R_{T_0}(\w)}\right|\supw\left| \frac{R_{T_0}(\w)}{R_{T_1}(\w)}\right|
\n\\
&\leq&\supw\left| \frac{L_{T_1}(\wh\w)}{L_{T_1}(\w)}\right|
\left\{ 1+\supw\left|\frac{L_{T_1}(\w)-R_{T_1}(\w)}{R_{T_0}(\w)}\right| \right\}\supw\left| \frac{R_{T_0}(\w)}{R_{T_1}(\w)}\right|
\\
&=&\supw\left| \frac{L_{T_1}(\wh\w)}{L_{T_1}(\w)}\right|
\left\{ 1+\supw\left|\frac{L_{T_1}(\w)-R_{T_1}(\w)}{R_{T_0}(\w)}\right| \right\}\left\{\infw\left| \frac{R_{T_1}(\w)}{R_{T_0}(\w)}\right|\right\}^{-1}
\n\\
&\leq&\supw\left| \frac{L_{T_1}(\wh\w)}{L_{T_1}(\w)}\right|
\left\{ 1+\supw\left|\frac{L_{T_1}(\w)-R_{T_1}(\w)}{R_{T_0}(\w)}\right| \right\}\left\{1-\supw\left| \frac{R_{T_0}(\w)-R_{T_1}(\w)}{R_{T_0}(\w)}\right|\right\}^{-1},\n
\ee
and
\be
\label{a1.5}
&&\xi_{T_1}^{-1}L_{T_1}(\wh\w)
\n\\
&=&\supw\left| \frac{L_{T_1}(\wh\w)}{L_{T_1}(\w)}\frac{L_{T_1}(\w)}{R_{T_0}(\w)} \frac{R_{T_0}(\w)}{R_{T_1}(\w)} \right|
\n\\
&\geq&\supw\left| \frac{L_{T_1}(\wh\w)}{L_{T_1}(\w)}\right|
\infw\left|\frac{L_{T_1}(\w)}{ R_{T_0}(\w)}\right|\infw\left|\frac{R_{T_0}(\w)}{R_{T_1}(\w)}\right|
\n\\
&\geq&\supw\left| \frac{L_{T_1}(\wh\w)}{L_{T_1}(\w)}\right|
\left\{ 1-\supw\left|\frac{L_{T_1}(\w)-R_{T_1}(\w)}{R_{T_0}(\w)}\right| \right\}\infw\left|\frac{R_{T_0}(\w)}{R_{T_1}(\w)}\right|
\\
&=&\supw\left| \frac{L_{T_1}(\wh\w)}{L_{T_1}(\w)}\right|
\left\{ 1-\supw\left|\frac{L_{T_1}(\w)-R_{T_1}(\w)}{R_{T_0}(\w)}\right| \right\}\left\{\supw\left|\frac{R_{T_1}(\w)}{R_{T_0}(\w)}\right|\right\}^{-1}
\n\\
&\geq&\supw\left| \frac{L_{T_1}(\wh\w)}{L_{T_1}(\w)}\right|
\left\{ 1-\supw\left|\frac{L_{T_1}(\w)-R_{T_1}(\w)}{R_{T_0}(\w)}\right| \right\}\left\{1+\supw\left|\frac{R_{T_0}(\w)-R_{T_1}(\w)}{R_{T_0}(\w)}\right|\right\}^{-1}.\n
\ee
Based on~\eqref{opt1}, \eqref{a0.2} and \eqref{a1.3}, these bounds suggest that $\xi_{T_1}^{-1}L_{T_1}(\wh\w)=1+o_p(1)$, which further leads to $\xi_{T_1}^{-1}\left\{L_{T_1}(\wh\w)-\xi_{T_1}\right\} = o_p(1)$.
Hence, based on the uniform integrability of  $\xi_{T_1}^{-1}\left\{L_{T_1}(\wh\w)-\xi_{T_1}\right\}$, we can show  \eqref{opt2} in Theorem~\ref{th:opt} for a divergent $T_1$ at rate $O(T_0)$. This completes the proof of Theorem~\ref{th:opt}. 
\hqed

\section{Proof of Theorem~\ref{th:opt with intercept}}
\label{sec:pf3}
\setcounter{equation}{0} 
The proof of~\eqref{opt0 with intercept} in Theorem~\ref{th:opt with intercept} resembles that of \eqref{opt0} in Theorem~\ref{th:opt}, namely, we need to show 
\be
\supwd \left| \frac{L_{T_0}(\w,d)-R_{T_0}(\w,d)}{ R_{T_0}(\w,d)} \right|=o_p(1),
\label{s0.1}
\ee
and
\be
\supwd \left| \frac{R_{T_0}(\w,d)-R_{T_1}(\w,d)}{ R_{T_0}(\w,d)} \right|=o(1).
\label{s0.2}
\ee
We first show~\eqref{s0.1}. We use the same notations 
	$\wt\blam_{t}$, $\wt\bmu_i$, $\e_{\wt\bmu}^{(i)}$ and $e_{t,\epsilon}^{(i)}$ defined as above. Under the linear factor structure~\eqref{DGP-factor model}, we have
\be
&&L_{T_0}(\w,d)\n\\
&=&\frac{1}{T_0}
\left\{\sum_{t\in\calT_0}\left(y_{0,t}-\sumj w_j y_{j,t}-d\right)^2+\left\|\Z_0-\sumj w_j\Z_j-d\biota_{r}\right\|^2\right\}
\n\\
&=&L_{T_0}(\w)+d^2
-\frac{2d}{T_0}\sumt\left(y_{0,t}-\sumj w_j y_{j,t}\right)
 +\frac{rd^2}{T_0}-\frac{2d}{T_0}\biota_{r}^\top\left( \Z_0-\sumj w_j\Z_j \right) 
\n\\
&=&L_{T_0}(\w)+\frac{d^2(T_0+r)}{T_0}-\frac{2d}{T_0}\sumt\sumj w_j\left\{\wt\blam_{t}^\top\e_{\wt\bmu}^{(j)}+e_{t,\epsilon}^{(j)} \right\} 
-\frac{2d}{T_0}\biota_{r}^\top\left( \Z_0-\sumj w_j\Z_j \right),
\n
\ee 
for any $\w\in\calH$ and $d\in\Ra$.
Then under Condition~\ref{as:data stochastic}, we have
\be
&&R_{T_0}(\w,d)
\n\\
&=&\E L_{T_0}(\w,d)
\n\\
&=&R_{T_0}(\w)+\frac{d^2(T_0+r)}{T_0}-\frac{2d}{T_0}\sumt\sumj w_j \wt\blam_{t}^\top\e_{\wt\bmu}^{(j)}
-\frac{2d}{T_0}\biota_{r}^\top\left( \Z_0-\sumj w_j\Z_j \right).
\label{s0.pre.Rwd}
\ee
This further suggests that 
\be
&&\supwd \left| \frac{L_{T_0}(\w,d)-R_{T_0}(\w,d)}{ R_{T_0}(\w,d)} \right| 
\n\\
&\leq&
\wt\xi_{T_0}^{-1}\supwd \left| L_{T_0}(\w,d)-R_{T_0}(\w,d) \right|
\n\\
&=&\wt\xi_{T_0}^{-1}\supwd\left|
L_{T_0}(\w)-R_{T_0}(\w)-\frac{2d}{T_0}\sumt\sumj w_j  e_{t,\epsilon}^{(j)}
\right|
\n\\
&\leq&\wt\xi_{T_0}^{-1}\supw\left|
L_{T_0}(\w)-R_{T_0}(\w)\right|+2C_d\wt\xi_{T_0}^{-1}\supw\left|\frac{1}{T_0}\sumt\sumj w_j  e_{t,\epsilon}^{(j)}\right|
\n\\
&\leq&\wt\xi_{T_0}^{-1}\supw\left|
L_{T_0}(\w)-R_{T_0}(\w)\right|+2C_d  \wt\xi_{T_0}^{-1}\supw\sumj |w_j|\left|\frac{1}{T_0}\sumt  e_{t,\epsilon}^{(j)}\right|
\n\\
&\leq&\wt\xi_{T_0}^{-1}\supw\left|
L_{T_0}(\w)-R_{T_0}(\w)\right|+2C_d\max\{C_\text{L}, C_\text{U}\}\wt\xi_{T_0}^{-1} T_0^{-1/2}\sumj \Phi_{T_0}(j),
\label{s0.1.1}
\ee
where  
$\Phi_{T_0}(j)=\left|T_0^{-1/2}\sumt  e_{t,\epsilon}^{(j)}\right|$.
Under Condition~\ref{as:data mixing} and Theorem~3.49 in \cite{white1984}, $\left\{e_{t,\epsilon}^{(j)}\right\}$ for $j\in\{1,\ldots,J\}$ is either an $\alpha$-mixing sequence with the mixing coefficient $\alpha=-r/(r-2)$ or a $\phi$-mixing sequence with the mixing coefficient $\phi=-r/(2r-1)$, $r\geq 2$.
Moreover,  $\var\left\{e_{t,\epsilon}^{(j)}\right\}$ is uniformly bounded for $j\in\{1,\ldots,J\}$, based on Condition~\ref{as:error term}~\eqref{as:moment of error term}.
Furthermore, using Condition~\ref{as:error term with intercept}~\eqref{as:var of pre error term with intercept} and Theorem~5.20 in \cite{white1984}, we can obtain that 
\be
\sumi \Phi_{T_0}(j)=O_p(J).\n
\ee
Due to Conditions~\ref{as:data stochastic}, \ref{as:data mixing}, and \ref{as:error term}~\eqref{as:moment of error term}--\eqref{as:expection of error term}, Formulas~\eqref{a0.1.1}--\eqref{a0.1.2} hold, and thus we have
\be
\wt\xi_{T_0}^{-1}\supw\left|
L_{T_0}(\w)-R_{T_0}(\w)\right|=O_p(\wt\xi_{T_0}^{-1}T_0^{-1/2}J^2).
\label{s0.1.3}
\ee
Formulas~\eqref{s0.1.1} and~\eqref{s0.1.3} together with Condition~\ref{as:xi1 with intercept} imply that~\eqref{s0.1} holds.

We then show~\eqref{s0.2}. Similar to the statement in \eqref{a2.6}, we can obtain that
\be
\supw\left|\biota_{r}^\top\left( \Z_0-\sumj w_j\Z_j \right)\right|
\leq\supw\left\| \biota_{r} \right\| \left\|  \Z_0-\sumj w_j\Z_j \right\|
=O_p\left(J^2\right),
\label{s0.1.4-}
\ee
where the last equality is guaranteed by Condition~\ref{as:covariates}~\eqref{Z} and that $r$ is a fixed number.
Similar to the analysis of \eqref{a2.4}--\eqref{a2.7}, we can show that
\be
\supw\left|
R_{T_0}(\w)-R_{T_1}(\w)\right|=O_p(T_0^{-1/2}J^2)+o(\wt\xi_{T_0}),
\label{s0.1.4}
\ee
under Conditions~\ref{as:data stochastic}--\ref{as:covariates} and \ref{as:post error term bias with intercept}.
Similar to \eqref{s0.pre.Rwd}, Conditions~\ref{as:factor loading} and \ref{as:factor bound with intercept} imply that
\be
R_{T_1}(\w,d)
=R_{T_1}(\w)+d^2-\frac{2d}{T_1}\sumT\sumj w_j \wt\blam_{t}^\top\e_{\wt\bmu}^{(j)}.
\n
\ee
Thus, we obtain that
\be
&&\supwd \left| \frac{R_{T_0}(\w,d)-R_{T_1}(\w,d)}{ R_{T_0}(\w,d)} \right|
\n\\
&\leq&\wt\xi_{T_0}^{-1}\supwd \left| R_{T_0}(\w,d)-R_{T_1}(\w,d) \right|
\n\\
&\leq& 2C_d\wt\xi_{T_0}^{-1}\supw\left|\sumj w_j\left(\frac{1}{T_0}\sumt\wt\blam_{t}-\frac{1}{T_1}\sumT\wt\blam_{t}\right)^\top\e_{\wt\bmu}^{(j)}\right|\n\\
&&+\wt\xi_{T_0}^{-1}\supw \left| R_{T_0}(\w)-R_{T_1}(\w) \right|
+{rC_d^2}\wt\xi_{T_0}^{-1}{T_0}^{-1}+{2C_d}\wt\xi_{T_0}^{-1}{T_0}^{-1}\supw\left|\biota_{r}^\top\left( \Z_0-\sumj w_j\Z_j \right)\right|
\n\\
&\leq& 2C_d\max\{C_{\text{L}},C_{\text{U}}\}\wt\xi_{T_0}^{-1}\sumj\left| \left(\frac{1}{T_0}\sumt\wt\blam_{t}-\frac{1}{T_1}\sumT\wt\blam_{t}\right)^\top\e_{\wt\bmu}^{(j)}\right|
\n\\
&&+\wt\xi_{T_0}^{-1}\supw \left| R_{T_0}(\w)-R_{T_1}(\w) \right|
+{rC_d^2}\wt\xi_{T_0}^{-1}{T_0}^{-1}+2C_d\wt\xi_{T_0}^{-1}{T_0}^{-1}\supw\left|\biota_{r}^\top\left( \Z_0-\sumj w_j\Z_j \right)\right|
\n\\
&=&
O\left(\wt\xi_{T_0}^{-1}T_0^{-1/2}J^2\right)+o(1),
\label{s0.2.1}
\ee
where the last equality holds due to \eqref{s0.1.4-}--\eqref{s0.1.4} together with Conditions~\ref{as:factor loading} and \ref{as:factor bound with intercept}~\eqref{lambda with intercept}.
Combining \eqref{s0.2.1} with Condition~\ref{as:xi1 with intercept}, we can obtain~\eqref{s0.2}. Thus,~\eqref{opt0 with intercept} holds.

Next, we show \eqref{opt1 with intercept} in Theorem~\ref{th:opt with intercept}. It suffices to show that
\be
\supwd \left| \frac{L_{T_1}(\w,d)-R_{T_1}(\w,d)}{R_{T_0}(\w,d)} \right|=o_p(1).
\label{s1.3}
\ee
This equation holds because when $T_1$ is divergent at rate $O(T_0)$, 
\be
\supwd \left| \frac{L_{T_1}(\w,d)-R_{T_1}(\w,d)}{ R_{T_0}(\w,d)} \right|=O_p(\wt\xi_{T_0}^{-1}T_0^{-1/2}J^2),\n
\ee
under Conditions~\ref{as:data stochastic}, \ref{as:data mixing}, \ref{as:error term}~\eqref{as:moment of error term}, \ref{as:error term}~\eqref{as:var of post error term} and \ref{as:error term with intercept}~\eqref{as:var of post error term with intercept}, similar to the analyses of \eqref{s0.1.1}--\eqref{s0.1.3}. 

Finally, we show~\eqref{opt2 with intercept} in Theorem~\ref{th:opt with intercept}. Similar to the arguments of~\eqref{a1.4}--\eqref{a1.5}, we can show that
\be
\wt\xi_{T_1}^{-1}L_{T_1}(\wh\w^{\text{DSC}},\wh d^{\text{DSC}})
&\leq&\supwd\left| \frac{L_{T_1}(\wh\w^{\text{DSC}},d)}{L_{T_1}(\w,d)}\right|
\left\{ 1+\supwd\left|\frac{L_{T_1}(\w,d)-R_{T_1}(\w,d)}{R_{T_0}(\w,d)}\right| \right\}
\n\\
&&\left\{1-
\supwd\left| \frac{R_{T_0}(\w,d)-R_{T_1}(\w,d)}{R_{T_0}(\w,d)}\right|\right\}^{-1},
\label{s1.4}
\ee
and
\be
\wt\xi_{T_1}^{-1}L_{T_1}(\wh\w^{\text{DSC}},\wh d^{\text{DSC}})
&\geq&\supwd\left| \frac{L_{T_1}(\wh\w^{\text{DSC}},d)}{L_{T_1}(\w,d)}\right|
\left\{ 1-\supwd\left|\frac{L_{T_1}(\w,d)-R_{T_1}(\w,d)}{R_{T_0}(\w,d)}\right| \right\}
\n\\
&&\left\{1+\supwd\left|\frac{R_{T_0}(\w,d)-R_{T_1}(\w,d)}{R_{T_0}(\w,d)}\right|\right\}^{-1}.
\label{s1.5}
\ee
By \eqref{opt1 with intercept}, \eqref{s0.2} and \eqref{s1.3}--\eqref{s1.5}, 
we obtain that $\wt\xi_{T_1}^{-1}L_{T_1}(\wh\w^{\text{DSC}},\wh d^{\text{DSC}})=1+o_p(1)$, which further implies that $\wt\xi_{T_1}^{-1}\left\{L_{T_1}(\wh\w^{\text{DSC}},\wh d^{\text{DSC}})-\wt\xi_{T_1}\right\} = o_p(1)$.
Hence, by the uniform integrability of  $\wt\xi_{T_1}^{-1}\left\{L_{T_1}(\wh\w^{\text{DSC}},\wh d^{\text{DSC}})-\wt\xi_{T_1}\right\}$, we show \eqref{opt2 with intercept} in Theorem~\ref{th:opt with intercept} when $T_1$ is divergent at rate  $O(T_0)$. This completes the proof of Theorem~\ref{th:opt with intercept}.
\hqed

\newpage
\begin{center}
	{\huge Online Appendix to \\\vspace*{1cm}
Asymptotic Properties of the Synthetic Control Method}\ \\
\end{center}
\vspace*{2cm}
\renewcommand{\thesection}{S.\arabic{section}}
\renewcommand{\thesubsection}{S.\arabic{section}.\arabic{subsection}}
\graphicspath{{figure/}}

\renewcommand{\thefigure}{S.\arabic{figure}}
\renewcommand{\thetable}{S.\arabic{table}}

\renewcommand{\theequation}{S.\arabic{equation}}

This file contains additional theoretical discussions and proofs. Specifically, Section~\ref{sec:pf4} provides further discussion on the sufficient condition of Assumption~3.2 in \cite{ferman2021}. Section~\ref{sec:relaxing DGP} relaxes the assumption of linear factor DGPs  and proves the convergence of SC weight and the asymptotic optimality of SCM in a model-free setup. 

All limiting processes in the theoretical results below correspond to $n\to \infty$ unless stated otherwise.

\renewcommand{\thecondition}{S\arabic{section}.\arabic{condition}}
\setcounter{condition}{0}
\setcounter{theorem}{3}
\setcounter{section}{0}

\section{Explanation of the sufficient condition of Assumption~3.2 in \cite{ferman2021}}
\label{sec:pf4}
\setcounter{equation}{0} 
In this section, we prove that \eqref{bound of MM} and \eqref{bound of Q} are sufficient to guarantee Assumption~3.2 of \citet{ferman2021}, i.e., $\left\| \bmu_{0}-\M\w_{T_1}^\text{opt} \right\|\to0$  and $\left\| \w_{T_1}^\text{opt} \right\|\to0$.

Denote $\{\Lambda_1,\ldots,\Lambda_J\}$ as the singular value of matrix $\M^\top\Q\M$, and use $\zeta_{min}(\cdot)$ and $\zeta_{max}(\cdot)$ to represent the minimum and maximum singular value of a matrix. From \eqref{bound of MM} and \eqref{bound of Q}, we have
\be
&& \min\{\Lambda_1,\ldots,\Lambda_J\} \n\\
&=& J\zeta_{min}\left({J}^{-1}\M^\top\Q\M\right) \n\\
&\geq& J\zeta_{min}\left({J}^{-1}\M\M^\top\right)\zeta_{min}(\Q) \n\\
&=& J\lambda_{min}\left({J}^{-1}\M\M^\top\right)\lambda_{min}(\Q)  \n\\
&\geq& c_1c_3J
\label{bound of Lambda}
\ee 
and
\be
\zeta_{max}(\M)=\sqrt{J\lambda_{max}\left(J^{-1}\M^\top\M\right)}=O(J^{1/2}).
\label{bound of M}
\ee
With the analytical solution of $\w_{T_1}^\text{opt}$, we can obtain that
\be
\bmu_{0}-\M\w_{T_1}^\text{opt}&=&\left\{\I_{F}- \M\left(\M^\top\Q\M+\sigma_{\epsilon}^2\I_{J}\right)^{-1}\M^\top\Q \right\}\bmu_{0}
\n\\
&&-\frac{1}{2}\M\left(\M^\top\Q\M+\sigma_{\epsilon}^2\I_{J}\right)^{-1}\left(\brho_1-\rho_2\biota_{J}\right).
\n
\label{miu-Mw_opt}
\ee
Let $\P= \left\{\I_{F}- \M\left(\M^\top\Q\M+\sigma_{\epsilon}^2\I_{J}\right)^{-1}\M^\top\Q \right\}$.

We first consider the case where $\brho_1-\rho_2\biota_{J}=0$. Using \eqref{bound of Lambda} and Condition~\ref{as:factor loading}, we have that
\be
&&\left\| \bmu_{0}-\M\w_{T_1}^\text{opt} \right\| \n\\
&=&\left\|\P\bmu_{0} \right\|	\n\\
&\leq& F^{1/2} C_0  \zeta_{max}\left(\P\right) \n\\
&\leq& F^{1/2} C_0 \left\{ 1-\zeta_{min}\left(\M\left(\M^\top\Q\M+\sigma_{\epsilon}^2\I_{J}\right)^{-1}\M^\top\Q \right) \right\} \n\\
&=& F^{1/2} C_0 \left\{ 1-\zeta_{min}\left(\left(\M^\top\Q\M+\sigma_{\epsilon}^2\I_{J}\right)^{-1}\M^\top\Q\M \right) \right\} \n\\
&=& F^{1/2} C_0 \left\{ 1-\frac{\min\{\Lambda_1,\ldots,\Lambda_J\}}{\min\{\Lambda_1,\ldots,\Lambda_J\}+\sigma_{\epsilon}^2} \right\} \n\\
&\leq&  \frac{F^{1/2} C_0\sigma_{\epsilon}^2}{c_1c_3J+\sigma_{\epsilon}^2}.
\label{bound of mu0-Mw}
\ee
Thus, $\left\| \bmu_{0}-\M\w_{T_1}^\text{opt} \right\|\to0$ holds as $J$ diverges. 
Moreover, $\brho_1-\rho_2\biota_{J}=0$ also implies that
$$
\w_{T_1}^\text{opt}  = \sigma_{\epsilon}^{-2} \M^\top\Q \left(\bmu_{0}-\M\w_{T_1}^\text{opt} \right).
$$
Thus, along with  \eqref{bound of Q}, \eqref{bound of M} and \eqref{bound of mu0-Mw}, we obtain that
\be
\left\|\w_{T_1}^\text{opt}\right\| 
\leq \sigma_{\epsilon}^{-2} \zeta_{max}\left( \M \right) \zeta_{max}\left(\Q\right) \left\|\bmu_{0}-\M\w_{T_1}^\text{opt} \right\|
=O\left( J^{-1/2}  \right).
\n
\ee
This implies that $\left\|\w_{T_1}^\text{opt}\right\| \to0$ as $J$ diverges. 

Next, we consider $\brho_1-\rho_2\biota_{J}\neq0$. In this case, one may need stronger conditions to guarantee $\left\|\w^{\text{opt}}_{T_1}\right\|\to0$ and $\left\|\bmu_0-\M\w^{\text{opt}}_{T_1}\right\|\to0$. For example, $\left\|\w^{\text{opt}}_{T_1}\right\|$ and $\left\|\bmu_0-\M\w^{\text{opt}}_{T_1}\right\|$ both converge to zero when $\lambda_{min}\left( \M^\top\Q\M+\sigma_{\epsilon}^2\I_{J} \right)$ diverges faster than $\zeta_{max}(\M)\left\|\brho_1-\rho_2\biota_{J}\right\|$ as $J$ increases, and this condition can be interpreted as an increasing amount of information in the factors and loadings as $J$ diverges. Intuitively, this condition means that the factor structure can be recovered if the newly added control units provide additional information on the factors and loadings.

\section{Relaxing the linear factor model DGP}
\label{sec:relaxing DGP}
In this section, we examine the asymptotic behavior of SC weights and  asymptotic optimality of SC estimators without assuming the DGP to be a linear factor model. 
For coherence, we list the theorems again.

\subsection{Convergence rate of the SC weight}
%
%
%
%

\begin{theorem} \label{th:con relax dgp}	
	Given any $T_1$, if $\w^{\text{opt}}_{T_1}$ is an interior point of $\calH$ and Conditions~\ref{as:covariates}~\eqref{Z}, \ref{as:Y_c} and \ref{as:data stochastic relax dgp}--\ref{as:risk bound} hold, {then} 
	\begin{equation}
		\left\|\wh\w-\w^{\text{opt}}_{T_1}\right\|=O_p\left( T_0^{\nu}\xi_{T_0}^{1/2}+T_0^{\nu}\xi_{T_1}^{1/2}+T_0^{-1/4+\nu}J\right),
		\n
	\end{equation}
	where $\nu>0$ is a sufficiently small constant, $\xi_{T_0}=\inf_{\w \in \calH}R_{T_0}(\w)$, and $\xi_{T_1}=\inf_{\w \in \calH}R_{T_1}(\w)$.
\end{theorem}
This theorem generalizes Theorem~\ref{th:con} in a model-free setup. 
\begin{proof}[Proof of Theorem~\ref{th:con relax dgp}]
	We follow the proof of Theorem~\ref{th:con} to verify Theorem~\ref{th:con relax dgp}. 
	Thus, it suffices to show that for any $\varepsilon>0$, there exists a constant $C_\varepsilon$, such that  \eqref{a2.1} holds, under a given $T_1$ and any sufficient large $T_0$, without assuming the DGP to be a linear factor structure. 
	Recall that $\tau_{T_0}= T_0^{\nu}\xi_{T_0}^{1/2}+T_0^{\nu}\xi_{T_1}^{1/2}+T_0^{-1/4+\nu}J$, 
	and $\u=(u_1,\ldots,u_J)^\top$ satisfying $\left\|\u\right\|=C_\varepsilon$ and $\left(\w^{\text{opt}}_{T_1}+\tau_{T_0}\u\right)\in\calH$.
	Based on the analysis in Appendix~\ref{sec:pf2}, we need to investigate the order of $\Delta_{1}=-2\tau_{T_0}T_0^{-1} \left(\Y_0-{\boldsymbol{\mathcal{Y}}_c}\w^{\text{opt}}_{T_1}\right)^\top{\boldsymbol{\mathcal{Y}}_c}\u$ and $\Delta_{2}=\tau_{T_0}^2T_0^{-1} \left\|{\boldsymbol{\mathcal{Y}}_c}\u\right\|^2$ and then verify that $\Delta_{2}$ asymptotically dominates $\Delta_{1}$.
	
	First, we consider the order of $\Delta_{2}$. Based on Condition~\ref{as:Y_c},
	we can obtain that \eqref{a2.2}--\eqref{a2.3} hold with probablity approaching to 1, without specifying the model for the DGP.
	Next, we consider the order of $\Delta_{1}$. From Conditions~\ref{as:covariates}~\eqref{Z} and \ref{as:data stochastic relax dgp}, we can also obtain \eqref{a2.6}. Then, with  \eqref{a2.6} and Condition~\ref{as:risk bound}, we obtain that
	\be
	&&\supw\left| R_{T_0}(\w)-R_{T_1}(\w) \right| \n\\
	&=&\supw\left| \frac{1}{T_0}\sumt\E\left(y^N_{0,t}-\sumj w_j y^N_{j,t}\right)^2 +\frac{1}{T_0} \left\| \Z_0-\sumj w_j\Z_j \right\|^2-R_{T_1}(\w) \right|
	\n\\
	&=& \supw\left| \frac{1}{T_0}\sumt\E\left(y^N_{0,t}-\sumj w_j y^N_{j,t}\right)^2 +\frac{1}{T_0} \left\| \Z_0-\sumj w_j\Z_j \right\|^2-R_{T_1}(\w) \right| \n\\
	&\leq& \supw\left| \frac{1}{T_0}\sumt\E\left(y^N_{0,t}-\sumj w_j y^N_{j,t}\right)^2 -R_{T_1}(\w) \right| + \supw \frac{1}{T_0} \left\| \Z_0-\sumj w_j\Z_j \right\|^2
	\n\\
	&=& O(T_0^{-1/2}J^2)+o(\xi_{T_0}).
	\label{a4.1}
	\ee
	Similar to the analyses of \eqref{a2.8}--\eqref{a2.9}, combined with \eqref{a2.2}, we can derive the order of $\Delta_{1}$ as
	\be
	\left|\Delta_{1}\right|&\leq&
	\frac{2\tau_{T_0}}{T_0}\left\|\Y_0-{\boldsymbol{\mathcal{Y}}_c}\w^{\text{opt}}_{T_1}\right\|
	\left\|{\boldsymbol{\mathcal{Y}}_c}\u\right\|
	\n\\
	&=&
	O_p(\tau_{T_0}\xi_{T_1}^{1/2})\left\|\u\right\|+O_p(\tau_{T_0}T_0^{-1/4}J) \left\|\u\right\|+o_p(\tau_{T_0}\xi_{T_0}^{1/2})\left\|\u\right\|.
	\n
	\ee
	This formula, together with~\eqref{a2.3}, indicates that $\Delta_{2}$ asymptotically dominates $\Delta_{1}$. 
	This completes the proof of Theorem~\ref{th:con relax dgp}.
\end{proof}

\subsection{Asymptotic optimality of SCM without an intercept}

\begin{theorem} \label{th:opt relax dgp}
	If $T_1$ is finite, then under Conditions~\ref{as:covariates}~\eqref{Z}, \ref{as:xi1}, \ref{as:data stochastic relax dgp}--\ref{as:data mixing relax dgp} and
	\ref{as:error term relax dgp}~\eqref{as:moment of error term relax dgp}--\eqref{as:var of pre error term relax dgp}, we have 	
	\begin{equation}
		\frac{ R_{T_1}(\wh \w) }
		{\inf_{\w \in \calH} { R_{T_1}(\w)} } \overset{p}{\rightarrow}1.
		\label{opt0 relax dgp}
	\end{equation}
	If $T_1$ diverges at rate $O(T_0)$, then under Conditions~\ref{as:covariates}~\eqref{Z}, \ref{as:xi1} and \ref{as:data stochastic relax dgp}--\ref{as:error term relax dgp}, we have \eqref{opt0 relax dgp} and
	\begin{equation}
		\frac{ L_{T_1}(\wh \w) }
		{\inf_{\w \in \calH} { L_{T_1}(\w)} } \overset{p}{\rightarrow}1;
		\label{opt1 relax dgp}
	\end{equation}
	furthermore, if $\xi_{T_1}^{-1}\left\{L_{T_1}(\wh\w)-\xi_{T_1}\right\}$ is uniformly integrable, then
	\begin{equation}
		\frac{\E L_{T_1}(\wh \w) }
		{\inf_{\w \in \calH} { R_{T_1}(\w)} } \rightarrow 1.
		\label{opt2 relax dgp}
	\end{equation}
\end{theorem}
The results in~\eqref{opt0 relax dgp}--\eqref{opt2 relax dgp} are the same in~\eqref{opt0}--\eqref{opt2} in Theorem~\ref{th:opt}, but relax the assumption on the linear factor model as the outcome process. 

\begin{proof}[Proof of Theorem~\ref{th:opt relax dgp}]
	We follow the proof of Theorem~\ref{th:opt} to verify Theorem~\ref{th:opt relax dgp}. Based on the analyses in Appendix~\ref{sec:pf1}, it suffices to show \eqref{a0.1}--\eqref{a0.2} in order to prove \eqref{opt0 relax dgp}. Hence, we first prove \eqref{a0.1}, without the assumption of a linear factor DGP. 
	Note that \eqref{a0.1.1} continues to hold regardless of whether the DGP is a linear factor model.
	With $e_{t,y^N}^{(i)}=y^N_{0,t}-y^N_{i,t}$ for $i\in\{1,\ldots,J\}$ and $t\in\calT_0\cup\calT_1$,
	we can rewrite $\Psi_{T_0}(i,j)$ in \eqref{a0.1.1} as
	\be
	\Psi_{T_0}(i,j)=\left| \frac{1}{\sqrt{T_0}}\sumt \left\{ e_{t,y^N}^{(i)} e_{t,y^N}^{(j)} - \E e_{t,y^N}^{(i)}e_{t,y^N}^{(j)} \right\}
	\right|.\n
	\ee
	Similar to the statements of \eqref{a0.1.2}, 
	from Condition~\ref{as:data mixing relax dgp} and Conditions~\ref{as:error term relax dgp}~\eqref{as:moment of error term relax dgp}--\eqref{as:var of pre error term relax dgp}, we can obtain that
	\be
	\sumi\sumj \Psi_{T_0}(i,j)=O_p(J^2).
	\label{a4.2.1}
	\ee
	Combining~\eqref{a0.1.1}, \eqref{a4.2.1} and Condition~\ref{as:xi1}, we can obtain~\eqref{a0.1}. 
	Next, we consider the proof of \eqref{a0.2}. With Conditions~\ref{as:covariates}~\eqref{Z}, \ref{as:data stochastic relax dgp} and~\ref{as:risk bound}, we obtain~\eqref{a4.1}; thus,  \eqref{a0.2} follows from~\eqref{a4.1} and Condition~\ref{as:xi1}.
	Note that the proof above is valid regardless of whether $T_1$ is finite or divergent. Therefore, we complete the proof of \eqref{opt0 relax dgp} under a fixed $T_1$, as well as when $T_1$ diverges.

	Next, we prove \eqref{opt1 relax dgp} when $T_1$ diverges at rate $O(T_0)$. From the arguments in the proof of~\eqref{opt1} in Appendix~\ref{sec:pf1}, we need to show \eqref{a1.3} without assuming the DGP to be a factor model. Similar to the arguments of~\eqref{a0.1.1} and \eqref{a4.2.1}, under Conditions~\ref{as:data mixing relax dgp}, \ref{as:error term relax dgp}~\eqref{as:moment of error term relax dgp} and \ref{as:error term relax dgp}~\eqref{as:var of post error term relax dgp} and a divergent $T_1$ at rate $O(T_0)$, we can obtain that
	\be
	\supw \left| \frac{L_{T_1}(\w)-R_{T_1}(\w)}{ R_{T_0}(\w)} \right|=O_p(\xi_{T_0}^{-1}T_0^{-1/2}J^2).\n
	\n
	\ee
	This equation together with Condition~\ref{as:xi1} leads to~\eqref{a1.3}, which further proves~\eqref{opt1 relax dgp}.

	Finally, we verify \eqref{opt2 relax dgp} when $T_1$ diverges at rate $O(T_0)$. As \eqref{a1.4}--\eqref{a1.5} continue to hold without a linear factor DGP,
	combined with~\eqref{a0.2} and \eqref{a1.3}, we can obtain that $\xi_{T_1}^{-1}L_{T_1}(\wh\w)=1+o_p(1)$, which further leads to $\xi_{T_1}^{-1}\left\{L_{T_1}(\wh\w)-\xi_{T_1}\right\} = o_p(1)$.
	Moreover, by the uniform integrability of  $\xi_{T_1}^{-1}\left\{L_{T_1}(\wh\w)-\xi_{T_1}\right\}$, we can prove  \eqref{opt2 relax dgp} for a divergent $T_1$ at rate $O(T_0)$. This completes the proof of Theorem~\ref{th:opt relax dgp}. 
\end{proof}

\subsection{Asymptotic optimality of SCM with an intercept}

\begin{condition}\ 
	\begin{enumerate}[(i)]
		\item \label{as:var of pre error term with intercept relax dgp}
		There exists a constant $\wt C_4$ such that $\var\left\{T_0^{-1/2}\sumt  e_{t,y^N}^{(j)}\right\}\geq\wt C_4>0$ for all $T_0$ sufficiently large and any $ j\in\{1,\ldots,J\}$.
		\item \label{as:var of post error term with intercept relax dgp}
		There exists a constant $\wt C_5$ such that  $\var\left\{T_1^{-1/2}\sumT  e_{t,y^N}^{(j)}\right\}\geq\wt C_5>0$ for all $T_1$ sufficiently large and any $j\in\{1,\ldots,J\}$.
	\end{enumerate}
	\label{as:error term with intercept relax dgp}	
\end{condition}

\begin{condition}
	\label{as:bound with intercept relax dgp}
	${T_1}^{-1}\sumT\left({T_0}^{-1}\sum_{k\in\calT_0}\E y^N_{i,k}-\E y^N_{i,t}\right)=O\left(T_0^{-1/2}\right)$  for $i\in\{0,1,\ldots,J\}$ uniformly.
\end{condition}

\begin{condition}
	\label{as:post error term bias with intercept relax dgp}
	$\wt\xi_{T_0}^{-1}\supw\left| {T_0}^{-1}\sumt\E\left(\sumj w_j  e_{t,y^N}^{(j)}\right)^2-{T_1}^{-1}\sumT\E\left(\sumj w_j  e_{t,y^N}^{(j)}\right)^2\right|=o(1)$.
\end{condition}

Condition~\ref{as:error term with intercept relax dgp}	guarantees the application of the central limit theorem for dependent processes, as Condition~\ref{as:error term with intercept}. Compared to Condition~\ref{as:error term with intercept}, this condition directly concerns the potential outcome $y^N_{i,t}$ to allow for a model-free setup. Condition~\ref{as:bound with intercept relax dgp} can be viewed as a stationary assumption that generalizes Conditions~\ref{as:covariates}~\eqref{Z} and \ref{as:factor bound with intercept}--\ref{as:post error term bias with intercept}. Condition~\ref{as:post error term bias with intercept relax dgp} is a generalized version of Condition~\ref{as:post error term bias with intercept}.

\begin{theorem} \label{th:opt with intercept relax dgp}
	If $T_1$ is finite and Conditions~\ref{as:covariates}~\eqref{Z}, \ref{as:intercept bound}, \ref{as:xi1 with intercept}--\ref{as:data stochastic relax dgp}, \ref{as:data mixing relax dgp}, \ref{as:error term relax dgp}~\eqref{as:moment of error term relax dgp}--\eqref{as:var of pre error term relax dgp}, \ref{as:error term with intercept relax dgp}~\eqref{as:var of pre error term with intercept relax dgp} and
	\ref{as:bound with intercept relax dgp}--\ref{as:post error term bias with intercept relax dgp} hold, then  
	\begin{equation}
		\frac{ R_{T_1}(\wh\w^{\emph{DSC}},\wh d^{\emph{DSC}}) }
		{\inf_{\w \in \calH,d\in\Ra} { R_{T_1}(\w,d)} } \overset{p}{\rightarrow}1.
		\label{opt0 with intercept relax dgp}
	\end{equation}
	If $T_1$ diverges at rate $O(T_0)$ and
	Conditions~\ref{as:covariates}~\eqref{Z}, \ref{as:intercept bound}, \ref{as:xi1 with intercept}--\ref{as:data stochastic relax dgp} and \ref{as:data mixing relax dgp}--\ref{as:post error term bias with intercept relax dgp} hold, then we have \eqref{opt0 with intercept relax dgp} and
	\begin{equation}
		\frac{ L_{T_1}(\wh \w^{\emph{DSC}},\wh d^{\emph{DSC}}) }
		{\inf_{\w \in \calH,d\in\Ra} { L_{T_1}(\w,d)} } \overset{p}{\rightarrow}1;
		\label{opt1 with intercept relax dgp}
	\end{equation}
	furthermore, if $\wt\xi_{T_1}^{-1}\left\{L_{T_1}(\wh\w^{\emph{DSC}},\wh d^{\emph{DSC}})-\wt\xi_{T_1}\right\}$ is uniformly integrable, then
	\begin{equation}
		\frac{\E L_{T_1}(\wh \w^{\emph{DSC}},\wh d^{\emph{DSC}}) }
		{\inf_{\w \in \calH,d\in\Ra} { R_{T_1}(\w,d)} } \rightarrow 1.
		\label{opt2 with intercept relax dgp}
	\end{equation}
\end{theorem}
The results in~\eqref{opt0 with intercept relax dgp}--\eqref{opt2 with intercept relax dgp} generalize~\eqref{opt0 with intercept}--\eqref{opt2 with intercept} in Theorem~\ref{th:opt with intercept}, respectively, in a model-free setup. 

\begin{proof}[Proof of Theorem~\ref{th:opt with intercept relax dgp}]
	We follow the proof of Theorem~\ref{th:opt with intercept} to verify Theorem~\ref{th:opt with intercept relax dgp}. 
	Based on the analysis in Appendix~\ref{sec:pf3}, to prove  \eqref{opt0 with intercept relax dgp}, it suffices to show that \eqref{s0.1}--\eqref{s0.2} hold. 
	First, we give the proof of \eqref{s0.1} without assuming a factor model DGP. 
	Note that
	\be
	&&L_{T_0}(\w,d)
	\n\\
	&=&\frac{1}{T_0}
	\left\{\sum_{t\in\calT_0}\left(y^N_{0,t}-\sumj w_j y^N_{j,t}-d\right)^2+\left\|\Z_0-\sumj w_j\Z_j-d\biota_{r}\right\|^2\right\}
	\n\\
	&=&L_{T_0}(\w)+\frac{d^2(T_0+r)}{T_0}-\frac{2d}{T_0}\sumt\sumj w_j e_{t,y^N}^{(j)}
	-\frac{2d}{T_0}\biota_{r}^\top\left( \Z_0-\sumj w_j\Z_j \right),
	\n
	\ee 
	for any $\w\in\calH$ and $d\in\Ra$; thus, with Condition~\ref{as:data stochastic relax dgp}, 
	\be
	&&R_{T_0}(\w,d)
	\n\\
	&=&\E L_{T_0}(\w,d)
	\n\\
	&=& R_{T_0}(\w)+\frac{d^2(T_0+r)}{T_0}-\frac{2d}{T_0}\sumt\sumj w_j \E e_{t,y^N}^{(j)}
	-\frac{2d}{T_0}\biota_{r}^\top\left( \Z_0-\sumj w_j\Z_j \right).
	\label{s1.pre.Rwd}
	\ee 
	Similiar to \eqref{s0.1.1}, we have that
	\be
	&&\supwd \left| \frac{L_{T_0}(\w,d)-R_{T_0}(\w,d)}{ R_{T_0}(\w,d)} \right| 
	\n\\
	&\leq&\wt\xi_{T_0}^{-1}\supwd\left|
	L_{T_0}(\w)-R_{T_0}(\w)-\frac{2d}{T_0}\sumt\sumj w_j  \left\{e_{t,y^N}^{(j)}-\E e_{t,y^N}^{(j)} \right\}
	\right|
	\n\\
	&\leq&\wt\xi_{T_0}^{-1}\supw\left|
	L_{T_0}(\w)-R_{T_0}(\w)\right|+2C_d\max\{C_\text{L}, C_\text{U}\}\wt\xi_{T_0}^{-1} T_0^{-1/2}\sumj \wt \Phi_{T_0}(j),
	\label{s1.1.1}
	\ee
	where  
	$\wt\Phi_{T_0}(j)=\left|T_0^{-1/2}\sumt  \left\{e_{t,y^N}^{(j)}-\E e_{t,y^N}^{(j)} \right\} \right|$.
	Based on Conditions~\ref{as:data mixing relax dgp}, \ref{as:error term relax dgp}~\eqref{as:moment of error term relax dgp} and \ref{as:error term with intercept relax dgp}~\eqref{as:var of pre error term with intercept relax dgp}, we can obtain that 
	\be
	\sumj \wt \Phi_{T_0}(j)=O_p(J).
	\label{s1.1.2}
	\ee
	Note that \eqref{a0.1.1} still holds without assuming a linear factor DGP.
	With Conditions~\ref{as:data mixing relax dgp} and \ref{as:error term relax dgp}~\eqref{as:moment of error term relax dgp}--\eqref{as:var of pre error term relax dgp}, \eqref{a4.2.1} holds; thus, combined with \eqref{a0.1.1}, we can obtain \eqref{s0.1.3}. Formulas~\eqref{s0.1.3} and \eqref{s1.1.1}--\eqref{s1.1.2}, together with Condition~\ref{as:xi1 with intercept}, guarantee that \eqref{s0.1} holds without assuming a linear factor DGP.
	
	Then, we consider \eqref{s0.2}. 
	From Condition~\ref{as:bound with intercept relax dgp}, we have 
	\be
	&&\supw\left|\sumj w_j\left(\frac{1}{T_0}\sumt \E e_{t,y^N}^{(j)} -\frac{1}{T_1}\sumT\E e_{t,y^N}^{(j)} \right) \right|
	\n\\
	&=&\supw\left|\sumj w_j\left\{\left(\frac{1}{T_0}\sumt \E y^N_{0,t} -\frac{1}{T_1}\sumT\E y^N_{0,t} \right)-\left(\frac{1}{T_0}\sumt \E y^N_{j,t} -\frac{1}{T_1}\sumT\E y^N_{j,t} \right)\right\} \right|
	\n\\ 
	&\leq&\supw\left|\sumj w_j\left(\frac{1}{T_0}\sumt \E y^N_{0,t} -\frac{1}{T_1}\sumT\E y^N_{0,t} \right)\right|
	\n\\
	&&+\supw\left|\sumj w_j\left(\frac{1}{T_0}\sumt \E y^N_{j,t} -\frac{1}{T_1}\sumT\E y^N_{j,t} \right) \right|
	\n\\
	&=&O(T_0^{-1/2}).
	\label{s1.2}
	\ee
	Since $r$ is fixed, from Condition~\ref{as:covariates}~\eqref{Z}, we obtain \eqref{s0.1.4-}. 
	Similar to the arguments in the proof of \eqref{a4.1},
	we can obtain that
	\be
	\supw\left|
	R_{T_0}(\w)-R_{T_1}(\w)\right|=O_p(T_0^{-1/2}J^2)+o(\wt\xi_{T_0}),
	\label{s1.1.4}
	\ee
	under Conditions~\ref{as:covariates}~\eqref{Z}, \ref{as:data stochastic relax dgp} and \ref{as:post error term bias with intercept relax dgp}.
	Since
	\be
	R_{T_1}(\w,d)
	=R_{T_1}(\w)+d^2-\frac{2d}{T_1}\sumT\sumj w_j \E e_{t,y^N}^{(j)},
	\n
	\ee
	with \eqref{s0.1.4-}, \eqref{s1.pre.Rwd} and \eqref{s1.2}--\eqref{s1.1.4}, we have that
	\be
	&&\supwd \left| \frac{R_{T_0}(\w,d)-R_{T_1}(\w,d)}{ R_{T_0}(\w,d)} \right|
	\n\\
	&\leq&\wt\xi_{T_0}^{-1}\supwd \left| R_{T_0}(\w,d)-R_{T_1}(\w,d) \right|
	\n\\
	&\leq& 2C_d\wt\xi_{T_0}^{-1}\supw\left|\sumj w_j\left(\frac{1}{T_0}\sumt \E e_{t,y^N}^{(j)} -\frac{1}{T_1}\sumT\E e_{t,y^N}^{(j)} \right) \right|+\wt\xi_{T_0}^{-1}\supw \left| R_{T_0}(\w)-R_{T_1}(\w) \right|\n\\
	&&+{rC_d^2}\wt\xi_{T_0}^{-1}{T_0}^{-1}+{2C_d}\wt\xi_{T_0}^{-1}{T_0}^{-1}\supw\left|\biota_{r}^\top\left( \Z_0-\sumj w_j\Z_j \right)\right|
	\n\\
	&\leq& 2C_d\max\{C_{\text{L}},C_{\text{U}}\}\wt\xi_{T_0}^{-1}\sumj\left| \frac{1}{T_0}\sumt \E e_{t,y^N}^{(j)} -\frac{1}{T_1}\sumT\E e_{t,y^N}^{(j)} \right| +\wt\xi_{T_0}^{-1}\supw \left| R_{T_0}(\w)-R_{T_1}(\w) \right|
	\n\\
	&&+{rC_d^2}\wt\xi_{T_0}^{-1}{T_0}^{-1}+2C_d\wt\xi_{T_0}^{-1}{T_0}^{-1}\supw\left|\biota_{r}^\top\left( \Z_0-\sumj w_j\Z_j \right)\right|
	\n\\
	&=&O\left(\wt\xi_{T_0}^{-1}T_0^{-1/2}J^2\right)+o(1).
	\label{s1.2.1}
	\ee
	Combining \eqref{s1.2.1} with Condition~\ref{as:xi1 with intercept}, we can obtain \eqref{s0.2} without assuming a linear factor DGP. Thus, \eqref{opt0 with intercept relax dgp} holds.
	
	Next, we prove \eqref{opt1 with intercept relax dgp}.
	Based on the arguments in the proof of \eqref{opt1 with intercept} in Appendix~\ref{sec:pf3}, we need to show \eqref{s1.3} without assuming a linear factor DGP.	
	Similar to the arguments of \eqref{s1.1.1}--\eqref{s1.1.2}, when $T_1$ is divergent at rate $O(T_0)$, 
	\be
	\supwd \left| \frac{L_{T_1}(\w,d)-R_{T_1}(\w,d)}{ R_{T_0}(\w,d)} \right|=O_p(\wt\xi_{T_0}^{-1}T_0^{-1/2}J^2)
	\label{s2}
	\ee
	holds under Conditions~\ref{as:data mixing relax dgp}, \ref{as:error term relax dgp}~\eqref{as:moment of error term relax dgp}, \ref{as:error term relax dgp}~\eqref{as:var of post error term relax dgp} and \ref{as:error term with intercept relax dgp}~\eqref{as:var of post error term with intercept relax dgp}. Combining \eqref{s2} with Condition~\ref{as:xi1 with intercept}, we can obtain \eqref{s1.3}. Thus,  \eqref{opt1 with intercept relax dgp} is verified.

	Finally, we consider the proof of \eqref{opt2 with intercept relax dgp} when $T_1$ diverges at rate $O(T_0)$. 
	As \eqref{s1.4}--\eqref{s1.5} hold without assuming a linear factor DGP, combined with~\eqref{s0.2}, \eqref{s1.3} and \eqref{opt1 with intercept relax dgp},
	we obtain $\wt\xi_{T_1}^{-1}L_{T_1}(\wh\w^{\text{DSC}},\wh d^{\text{DSC}})=1+o_p(1)$, which further implies $\wt\xi_{T_1}^{-1}\left\{L_{T_1}(\wh\w^{\text{DSC}},\wh d^{\text{DSC}})-\wt\xi_{T_1}\right\} = o_p(1)$.
	Hence, by the uniform integrability of  $\wt\xi_{T_1}^{-1}\left\{L_{T_1}(\wh\w^{\text{DSC}},\wh d^{\text{DSC}})-\wt\xi_{T_1}\right\}$, we show \eqref{opt2 with intercept relax dgp} in Theorem~\ref{th:opt with intercept relax dgp} when $T_1$ is divergent at rate  $O(T_0)$. This completes the proof of Theorem~\ref{th:opt with intercept relax dgp}.
\end{proof}

\end{document}